\documentclass[11pt]{scrartcl}

\DeclareOldFontCommand{\bf}{\normalfont\bfseries}{\mathbf}

\usepackage{amsmath}
\usepackage[semicolon]{natbib}
\usepackage[pdftex, final]{graphicx}
\DeclareGraphicsExtensions{.png,.pdf,.jpg}
\usepackage{enumerate}
\usepackage{natbib}
\usepackage{url} 
\usepackage{soul}
\usepackage{subcaption}
\usepackage{amsmath}
\usepackage{bm}
\usepackage{verbatim}
\usepackage{amssymb, amsthm, amsmath}
\usepackage{booktabs}
\usepackage{appendix}
\usepackage{tcolorbox}
\usepackage[ruled,vlined]{algorithm2e}
\usepackage{longtable}
\usepackage{afterpage}
\usepackage{pdflscape}
\usepackage{capt-of}
\usepackage{lscape}

\usepackage{mathtools}
\usepackage[maxfloats=25]{morefloats}
\usepackage[export]{adjustbox}
\usepackage{subcaption}
\usepackage{dsfont}
\usepackage{placeins}

\newcommand{\blind}{1}

\newtheorem{theorem}{Theorem}
\newtheorem{proposition}{Proposition}

\newcommand{\E}[1]{\mathbb{E}\left[#1\right]}
\newcommand{\Prob}[1]{\mathbb{P}\left[ #1 \right]}
\newcommand{\ot}[2]{\tilde{\omega}_{#1  #2}}

\DeclarePairedDelimiter\floor{\lfloor}{\rfloor}

\usepackage{mathtools}
\usepackage[colorlinks=true,linkcolor=blue,citecolor=blue]{hyperref}%

\makeatletter
\newcommand*{\addFileDependency}[1]{
  \typeout{(#1)}
  \@addtofilelist{#1}
  \IfFileExists{#1}{}{\typeout{No file #1.}}
}
\makeatother

\begin{document}

\def\spacingset#1{\renewcommand{\baselinestretch}%
{#1}\small\normalsize} \spacingset{1}


\if1\blind
{
  
  \title{\bf A Common Atom Model for the Bayesian Nonparametric Analysis of Nested Data}

\author{ Francesco Denti\thanks{During the development of this article, F. Denti was also supported as a Ph.D. student by University of Milano - Bicocca, Milan, Italy and Universit\`a della Svizzera italiana, Lugano, Switzerland.}\\Department of Statistics\\
University of California, Irvine, CA \\ and\\  
Federico Camerlenghi\thanks{Also affiliated to Collegio Carlo Alberto, Piazza V. Arbarello 8, Torino and BIDSA, Bocconi University, Milano, Italy.}\\
Department of Economics, Management and Statistics\\ 
University of Milano - Bicocca, Milan, Italy\\
and\\ Michele Guindani \\
Department of Statistics\\
University of California, Irvine, CA \\ and \\ Antonietta Mira \\ Universit\`a della Svizzera italiana, Lugano, Switzerland \\Universit\`a dell'Insubria, Varese, Italy}
\maketitle
} \fi

\if0\blind
{
  \bigskip
  \bigskip
  \bigskip
  \begin{center}
    {\LARGE\bf A Common Atom Model for the Bayesian Nonparametric Analysis of Nested Data}
\end{center}
  \medskip
} \fi

\clearpage
\begin{abstract}
The use of high-dimensional data for targeted therapeutic interventions requires new ways to characterize the heterogeneity observed across subgroups of a specific population.  In particular, models for partially exchangeable data are needed for inference on nested datasets, where the observations are assumed to be organized in different units and some sharing of information is required to learn distinctive features of the units. In this manuscript, we propose a nested Common Atoms Model (CAM) that is particularly suited for the analysis of nested datasets where the distributions of the units are expected to differ only over a small fraction of the observations sampled  from each  unit.  The proposed CAM allows a two-layered clustering at the distributional and observational level and is amenable to scalable posterior inference through the use of a computationally efficient nested slice-sampler algorithm. We further discuss how to extend the proposed modeling framework to handle discrete measurements, and we conduct posterior inference on a real microbiome dataset from a diet swap study to investigate how the alterations in intestinal microbiota composition are associated with different eating habits. We further investigate the performance of our model in capturing true distributional structures in the population by means of a simulation study.
\end{abstract}

\noindent%
\vfill

\newpage
\spacingset{1.5} 
\section{Introduction}
\label{sec:intro}

The use of high-dimensional data for targeted therapeutic interventions requires new ways to characterize the heterogeneity observed across subgroups of a specific population.  In particular, models for partially exchangeable data are needed for inference on nested datasets, where the observations are assumed to be organized in different, though related, units. 
The borrowing of strength across units induced by these probabilistic structures is tailored to several applied problems. Here we deal with a  microbiome dataset made up of count measurements for 38 subjects (units) from both the U.S.A. and rural Africa, and the interest is to describe the different patterns of microbial diversity observed across the individuals since those patterns could inform future nutritional interventions.  The description of microbial diversity requires investigating the structure, concentration, and richness of microbiota in each subject and how the distributions of microbiota abundances vary across subgroups of subjects.  As the groups are typically unknown, they need to be estimated from the data. 

A few approaches have been proposed in the literature for clustering distributional features directly. For example, \citet{Irpino2015} have recently proposed clustering methods in symbolic statistics, by employing the Wasserstein distance on histograms treated as units. Similarly, \citet{Batagelj2015} have proposed generalized leaders and Ward’s hierarchical methods to cluster modal valued symbolic data. These are exploratory tools, which extend usual multivariate clustering methods to the analysis of (empirical) probability distributions, but they do not allow for a probabilistic assessment of cluster uncertainty.   

The Nested Dirichlet process \citep[nDP,][]{Rodriguez2008} and its extensions have been widely employed to identify distributional groups in Bayesian nonparametric model-based approaches. For example, \citet{Rodriguez2014} have proposed a generalization of the nDP for functional data analysis; \citet{Graziani2015} have investigated how the distribution of the changes of a targeted biomarker varies due to treatment and whether it is associated with a clinical outcome;  \cite{Zuanetti2018} have discussed a marginal nDP for clustering genes related to DNA mismatch repair via the distribution of gene-gene interactions with other genes. The nDP leads to a two-layered clustering: first, it allows grouping together similar units (distributional clustering), and then, within each distributional cluster, it clusters similar observations (observational clustering). 
However, \citet{Camerlenghi2018} have recently proved that the inference obtained using the nDP may be affected by a \emph{degeneracy} property: if two distributions share even only one atom in their support, the two distributions are automatically assigned to the same cluster.  
To overcome this drawback, \citet{Camerlenghi2018} propose a class of latent nested processes, which relies on estimating a latent mixture of shared and idiosyncratic processes across the subgroups. However, the computational burden of the resulting sampling scheme becomes demanding when the number of units increases. 

The degeneracy of the nDP is particularly problematic when analyzing high-dimensional data in genomics and microbiome studies. Here, the distribution profiles of sequencing data are expected to be quite similar across individuals and to vary only for a small fraction of differentially abundant sequences, which directly intervene to regulate the biological processes and their dysfunctions. 
Figure \ref{fig:histos} reports a snapshot of the observed microbial distributions for two representative individuals from the dataset we analyze in Section \ref{Sec:MicrobApp}.  In addition to the typical skewness and zero-inflation of microbial distributions, we note that the two distributions considerably overlap, and they are quite similar except for the presence of a small set of sequences which appear with high frequency.
In those applications, the nDP may provide unreliable inferences when comparing distributional patterns across individuals. 


In this paper, we propose a nested Common Atoms Model (CAM) that is particularly suited for the analysis of nested data sets, where the distributions of the units are expected to differ only over a small fraction of the observations. Although our proposal could be described as a constrained modification of the nDP, where atoms are allowed to be shared across all subgroups, the CAM i) does not suffer from the degeneracy issue of the nDP, and ii) allows scalable inference  with high-dimensional data. Furthermore, in the nDP, unit-level measurements can be clustered together only within units that are assigned to the same group. Thus, while the within-group clustering still contributes to a compact representation of the data, unit-level inference across subgroups is precluded. Instead, the proposed CAM framework naturally allows unit-level inference and clustering of observations across groups, since the structure of the common atoms allows mapping group-specific distributional patterns to a shared support. Compared to the proposal of \citet{Camerlenghi2018}, the proposed CAM is computationally more efficient, as it allows to conduct inference on a  larger number of observations and population subgroups. To this purpose, we develop a novel nested slice sampler algorithm \citep{Kalli2011}, which allows to target the true posterior distribution, without employing the standard truncation-based approximation, which is typically used for posterior inference with nDP models.

In the microbiome literature, ad-hoc solutions are sometimes adopted to address the challenges put forward by the analysis of microbiome data. For example, when dealing with the excess of zero counts, some authors simply add a small number (e.g. 1) to each count, thus generating ``pseudo counts". Here, we embed the proposed CAM framework within a rounded mixture of Gaussian (RGM) model \citep{Canale2011}. In this way, we effortlessly obtain a BNP nested model for count data that can naturally handle the sparsity and the zero-inflation typical of microbiome abundance tables. The resulting discrete CAM allows to cluster rows of an abundance table according to their distributional characteristics, providing a partition of patients with similar microbiome distribution. For example, the proposed CAM assigns the two subjects of Figure \ref{fig:histos} to two different population subgroups with high probability.

\begin{figure}[!t]
\begin{center}
  \begin{subfigure}[b]{0.45\textwidth}
    \includegraphics[width=\textwidth]{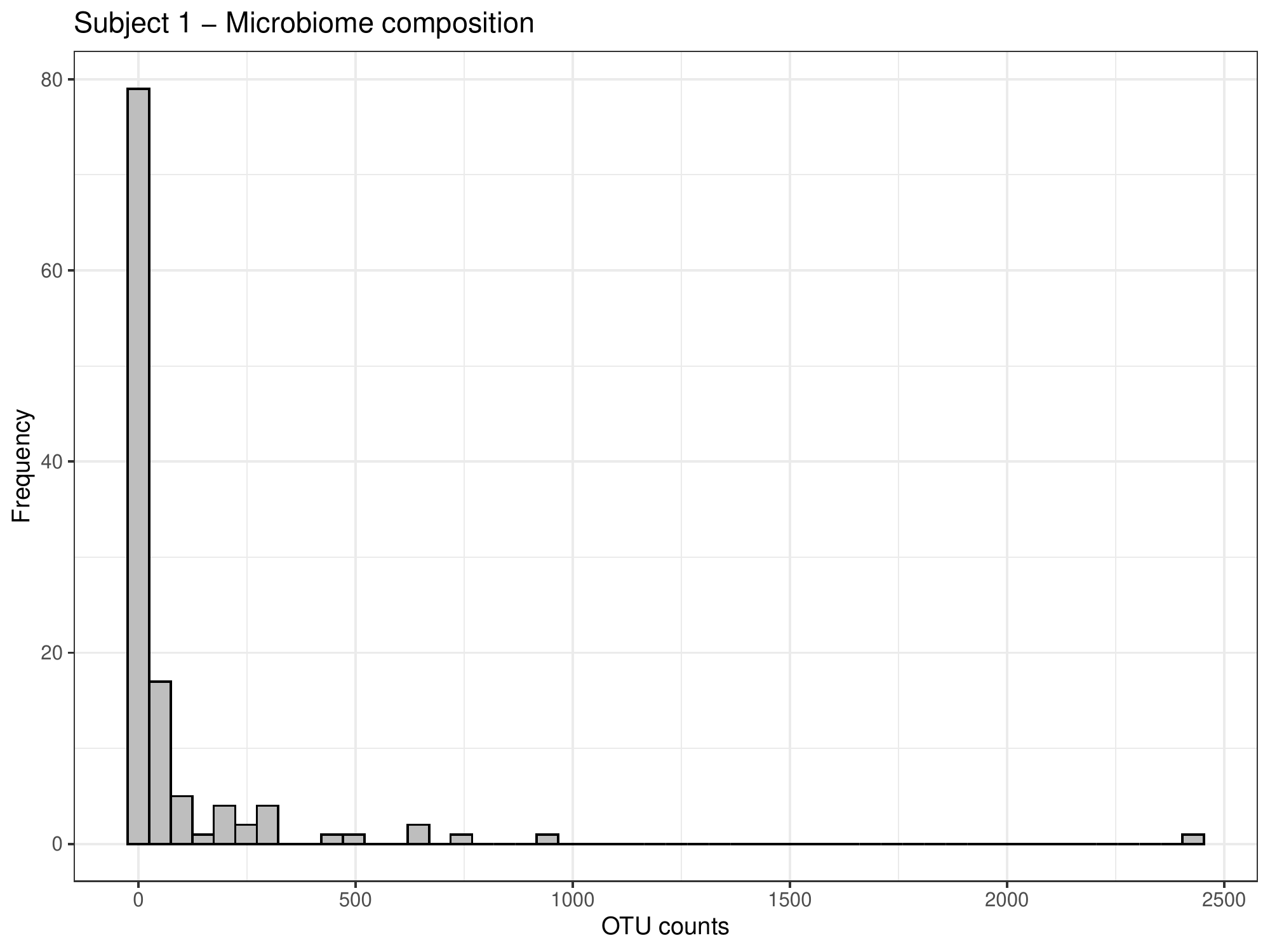}
    \caption{Subject 1}
    \label{fig:histos1}
  \end{subfigure}
  \begin{subfigure}[b]{0.45\textwidth}
    \includegraphics[width=\textwidth]{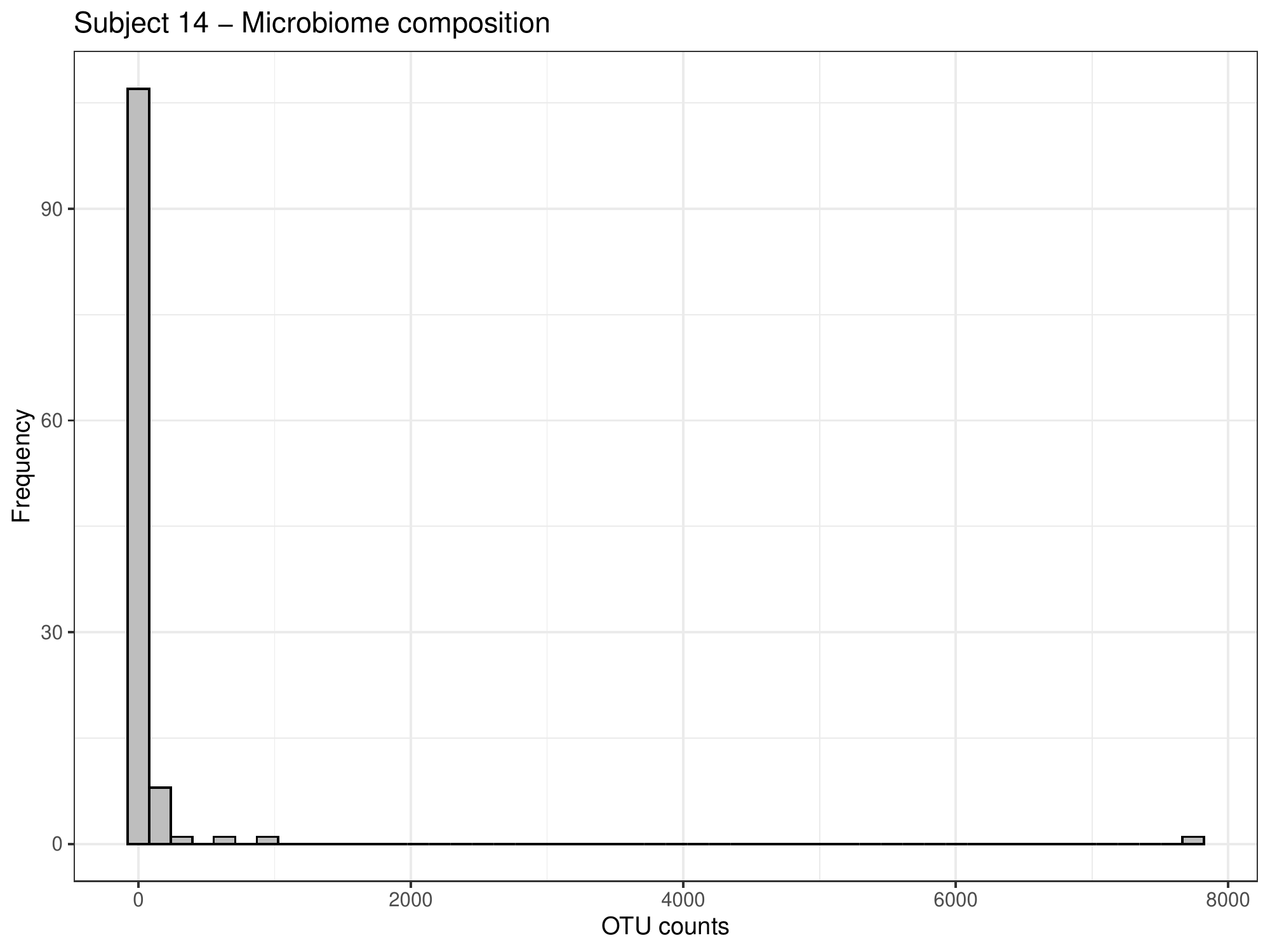}
    \caption{Subject 14}
    \label{fig:histos2}
  \end{subfigure}
  \caption{Histograms of the microbiome populations of two subjects in the study of \citet{OKeefe2015}. The distributions of the two units appear very similar and extremely skewed.}
    \label{fig:histos}
  \end{center}
\end{figure}

The remainder of the article is as follows. In Section \ref{SEC::2cam_continuous} we introduce our model for continuous measurements, and we discuss its properties. In Section \ref{Sec3:CAM_count} we discuss how to adapt the model to count data. In Section \ref{SEC::4Post}, we face posterior inference and outline the nested version of the slice sampler. Section \ref{Sec:MicrobApp} applies our model to a publicly available microbiome dataset in a diet swap study. Section \ref{SEC::6SimStudy} presents a simulation study to assess the clustering behavior of the model as the number of observations and groups grow in different scenarios. Section \ref{SEC::7Conclusion} summarizes our contributions and discusses some future directions. We defer proofs, additional algorithms, and simulation studies to the Supplementary Material.

\section{Common Atoms Model for Continuous Measurements}
\label{SEC::2cam_continuous}
We consider a \emph{nested} dataset, where we are provided with continuous measurements $\bm{y}_j= (y_{1,j}, \ldots , y_{n_j, j})$  observed over $J$ experimental units. We assume that each observation $y_{i,j}$,  $i=1,\ldots,n_j$ and $j=1,\ldots,J$, takes values in a suitable Polish space $\mathbb{X}$ endowed with the respective Borel $\sigma$-field $\mathcal{X}$. 
Similarly as in the nDP \citep{Rodriguez2008}, our goal is to achieve a partition of the vectors  $\bm{y}_1, \ldots, \bm{y}_J$ into a few, say $K \leq J$, distributional clusters. However, 
\citet{Camerlenghi2018} have shown that the partially exchangeable partition probability function of the nDP implies that distributions collapse into a common cluster when they share even only one atom. This unappealing behavior can be avoided if the prior explicitly models the commonality of atoms between groups. Here, we propose a Common Atoms Model (CAM) such that distributions belonging to different clusters are characterized by specific weights assigned to a common set of atoms. In this section, we define the model and investigate its properties for analyzing high-dimensional data. More specifically, let $G_j$, as $ \: j=1, \ldots, J$, denote the distribution of the $j$-th experimental unit, so that
\begin{equation}
\begin{aligned}
(y_{i_1,1},\ldots, y_{i_J, J}) |    \, G_1,\ldots,G_J &\overset{ind.}{\sim} G_1\times\cdots\times G_J\:, \qquad 
(i_1, \ldots , i_J) \in \mathbb{N}^J. 
\end{aligned}
\label{CAM1}
\end{equation}
Then, similarly as in the nDP formulation,  we assume that the $G_j$'s are a sample from an almost surely discrete distribution $Q$ over the space of probability distributions on $\mathcal{X}$, namely
\begin{equation}
\begin{aligned}
G_1,\ldots,G_J|Q &\overset{i.i.d.}{\sim}Q, \qquad \qquad Q  &=\sum_{k\geq 1}\,  \pi_k \, \delta_{G^*_k }.
\end{aligned}
\label{CAM2}
\end{equation}
\noindent
where $G^*_k = \sum_{l\geq 1}, \omega_{ l,k}\, \delta_{\theta_l}$, $ k \geq 1$, and the common atoms $\theta_1,\theta_2,\ldots$ are drawn from a non-atomic base measure $H$ on $(\mathbb{X}, \mathcal{X})$. We further assume the Griffiths-Engen-McCloskey (GEM) distribution for the weights, which characterizes the stick-breaking (or Sethuraman's)   construction of the Dirichlet process \citep{Sethuraman1994}, i.e. we consider $V_k\sim Beta(1,\alpha)$, $k\geq 1$, and then set  $\pi_1=V_1$, and $\pi_k=V_k\, \prod_{r=1}^{k-1} (1-V_r)$, $k>1$,  indicated as $\bm{\pi} = \{\pi_k\}_{k\geq 1}\sim GEM(\alpha)$. Similarly, $\bm{\omega}_k=\{\omega_{ l,k }\}_{l \geq 1}\sim GEM(\beta)$ for all $k\geq 1$.\\ 
The distribution defining $G_k^*$ can be seen as a single-atom dependent DP 
as defined in \citet[Definition 3]{barrientos2012}, indexed by a categorical covariate with support on $\mathbb{N}$.
\citet{Hatjispyros2016} have previously investigated the use of a common atoms structure to  model pairwise-dependent Dirichlet processes across $m$ \emph{known} sub-populations. Our CAM similarly employs common atoms to induce dependence across the $G_k^*$'s, but further allows clustering of distributional units, leading to a new model of nested random probability measures.
Due to the commonality of the atoms at the unit level, our construction is also reminiscent of the Hierarchical Dirichlet process (HDP) by \citet{Teh2006}. However, there are crucial differences between the two constructions. More specifically, the HDP does allow a flexible representation of each unit-level distribution $G_j$, $j=1, \ldots, J$, but does not induce distributional clusters among the units. Our formulation preserves a two-layered clustering structure, across units (distributional clustering) and between observations 
(observational clustering). Thus, the proposed CAM  is closer in spirit to recently developed hierarchical topic models, where an HDP is adopted as a base measure of an (outer) DP, in symbols $Q\sim DP(\alpha, HDP(\beta, H))$ \citep{Paisley2015, Tekumalla2015}. However, those nested HDP formulations aim at describing topic distributions which can be obtained as mixtures of separate topics (i.e. a document may contain words typical of both medicine and sports news), whereas our objective is to cluster individual distributions and the observations wherein (a patient-specific distribution is not obtained as a mixture of other patients' distributions). Hence, our proposal closely mimics the intended purpose of the original nDP model.
Finally, we mention an alternative semi--parametric model recently developed by \citet{Mario20} that also avoids the degeneracy issue of the nDP and allows for distributional clustering by extending the hierarchical Dirichlet process of \citet{Teh2006}. With respect to the work by \citet{Mario20}, our proposal is fully nonparametric, yet computationally efficient, and it easily accommodates extensions to the clustering of count data.

\subsection{Partition structure and correlation}
\label{SEC::21PartitionAndCorrelation}
In the following, we investigate some important properties of the proposed CAM  in terms of partition structure and correlation across groups. In particular, we show how the model does not suffer from the theoretical degeneracy of the nDP. We also discuss the implied dependence between pairs of observations and distributions. 

The discreteness of the random probability measures in our model \eqref{CAM1}--\eqref{CAM2} induces ties at the observational level, whose corresponding partition can be described via the so-called partially Exchangeable Partition Probability Function (pEPPF) (see, e.g., \citet{CamerlenghiAOS} and references therein). For notational simplicity, we illustrate the main results by focusing on  $J=2$, but our strategy easily extends to the general case. We further assume that there are $s >0$ distinct values out of a sample $\bm{y}_1, \ldots , \bm{y}_J$, which will be denoted by $y_1^*, \ldots , y_s^*$, with corresponding frequencies $\bm{n}_j=(n_{1,j}, \cdots, n_{s,j})$, where $n_{i,j}$ indicates the number of times that the $i$-th distinct value $y_i^*$ has been observed out of the initial sample in unit $j$.
We denote by $\mathsf{P}_\mathbb{X}$ the space of all random probability measures on $\mathbb{X}$. Our first result 
characterizes the mixed moments of the random probability measures $G_1$ and $G_2$ as a convex combination of the fully exchangeable case and a situation of independence across samples \citep[see also Proposition 2 in][]{Camerlenghi2018}. 
\begin{proposition} \label{prp:mixed_moments}
    Let $f_1$ and $f_2$ be two measurable functions defined on $\mathsf{P}_\mathbb{X}$ and taking values in 
    $\mathbb{R}^+$, then
    \begin{equation}
    \label{eq:mixed}
    \mathbb{E} \left[ \int_{\mathsf{P}_\mathbb{X}^2} f_1(g_1 ) f_2 (g_2) Q(d g_1) Q (d g_2) \right] =
    q_1 \mathbb{E} [f_1 (G_1^*)f_2 (G_1^*)] + (1-q_1) \mathbb{E} [f_1 (G_1^*) f_2 (G_2^*)]
    \end{equation}
    where we have set $q_1:= \mathbb{P} (G_1=G_2)$.
\end{proposition}
Following  \citet{CamerlenghiAOS}, we formally define the pEPPF as the probability of the observed allocation $\{\bm{n}_1,\ldots , \bm{n}_J\}$  of $s>0$ distinct observations out of the available sample, i.e.
\begin{equation}
\label{eq:pEPPF_def}
\Pi_N^{(s)} (\bm{n}_1,\ldots , \bm{n}_J) := \mathbb{E}\int_{\mathbb{X}^s} \prod_{j=1}^J \prod_{i=1}^s
G_j^{n_{i,j}} (d y_i^*) ,
\end{equation}
with $N= \sum_{j=1}^J n_j$. We point out that the \(i\)-th distinct value is shared by any two units \(j\) and \(\kappa\) if and only if \(n_{i, j}\, n_{i,\kappa} \geq 1\). If $J=1$ one obtains the usual exchangeable partition probability function (EPPF) for an individual sample, defined by \citep{Pitman1995}, and denoted here as 
$\Phi_{n_j}^{(s)}(\bm{n}_j)$. In the case of the Dirichlet process, this coincides with the well--known Ewens sampling formula,
\(
\Phi_{n_j}^{(s)}(\bm{n_j}) = \frac{\alpha^s \Gamma (\alpha)}{\Gamma (\alpha +n_j)} \prod_{i=1}^s (n_{i,j}-1)!
\)  \citep{Ewens72}. 
The pEPPF for the CAM  is described by the following theorem, for the case $J=2$.
\begin{theorem}
    \label{thm:EPPF}
    Let $\bm{y}_1$ and  $\bm{y}_2$ be samples from $J=2$ experimental units under the CAM  \eqref{CAM1}--\eqref{CAM2}. Then, the induced random partition of $s>0$ distinct observations  may be expressed as
    \begin{equation}
    \label{eq:pEPPF_CAM}
    \Pi_N^{(s)} (\bm{n}_1, \bm{n}_2) = q_1 \Phi_{n_1+n_2}^{(s)} (\bm{n}_1 +\bm{n_2}) +(1-q_1)
    \int_{\mathbb{X}^s} \mathbb{E} \prod_{j=1}^2 \prod_{i=1}^s (G_j^*)^{n_{i,j}} (d y_i^*).
    \end{equation}
\end{theorem}
Although a closed form expression is not available,  due to the presence of the integral over  $\mathbb{X}^s$ on the right hand side, the result is  fundamental to show  that the proposed CAM does not reduce to the fully exchangeable case in the presence of common observations across the two samples. Indeed,  we can prove the following:
\begin{proposition}  \label{prp:noEX}
    Assume that two samples $\bm{y}_1$ and $\bm{y}_2$ share $s_0 >0$ distinct observations, then
    \[
    \int_{\mathbb{X}^s} \mathbb{E} \prod_{j=1}^2 \prod_{i=1}^s (G_j^*)^{n_{i,j}} (d y_i^*) >0 .
    \]
\end{proposition} 
Theorem \ref{thm:EPPF} and Proposition \ref{prp:noEX} clarify that the pEPPF \eqref{eq:pEPPF_CAM} of our proposal does not reduce to the EPPF of the full exchangeable model. The proofs of the previous results are deferred to the Supplementary Material, where we also provide an explicit expression for the integral in \eqref{eq:pEPPF_CAM} (see Equation \eqref{eq:2term_pEPPF}).

Of course, ties among distributions at the outer level are still possible in view of the discreteness of $Q$ in \eqref{CAM2}. Indeed, if $j \not = j'$ we have
    \begin{equation}
    \begin{aligned}
    \mathbb{P}\left(G_j=G_{j'}|Q\right) = \sum_{k \geq 1} \pi_k^2>0, \quad \text{ and } \quad
    \mathbb{P}\left(G_j=G_{j'}\right) = \frac{1}{1+\alpha}.
    \end{aligned}
    \label{coclustG}
    \end{equation}
    Moreover, the probability of a tie between two data points in two separate units $j$ and $j'$, with $j \not = j'$, can be computed as
    \begin{equation}
    \begin{aligned}
    \Prob{y_{i,j} =y_{i',j'} }&=\frac{1}{1+\alpha}\left[\frac{1}{1+\beta} + \alpha\frac{1}{2\beta+1}\right]. 
        \end{aligned}
    \label{xijxij}
    \end{equation}
    This shows that CAM induces a two-fold clustering structure: it clusters together experimental units characterized by similar distribution profiles, and it also clusters together observations, allowing for borrowing information across the two layers.
    The determination of \eqref{coclustG}--\eqref{xijxij} is also deferred to the Supplementary Material.\\
    We conclude this section providing an explicit expression of the correlation between $G_j$ and $G_{j'}$ on different Borel sets, as $j \not = j'$; the covariance and correlation are useful quantities to  investigate the dependence across random probability measures and their suitability for practical applications. For any two Borel sets $A,B \in \mathcal{X}$ one has
\begin{equation} \label{eq:cov}
\begin{split}
& Cov\left(G_j(A),G_{j'}(B)\right) \\
& \qquad= H(A \cap B)\left( \frac{q_1}{1+\beta}+\frac{1-q_1}{1+2\beta} \right)-
H(A) H (B) \left( \frac{q_1}{1+\beta}+\frac{1-q_1}{1+2\beta} \right), 
\end{split}
\end{equation}
where $q_1=(1+\alpha)^{-1}$. In particular the correlation on the same set $A$ equals
\begin{equation} \label{eq:corr_same}
    \rho_{j,j'}:=Corr(G_j(A),G_{j'}(A)) = 1 - \frac{\beta}{2\beta+1}\,\frac{\alpha}{1+\alpha}.
\end{equation}
See Section \ref{Proof} of the Supplementary Material for the derivation of \eqref{eq:cov} and \eqref{eq:corr_same}.
It is interesting to note that $\rho_{j,j'}\in\left(1/2,1\right)$, due to the commonality of the atoms. In many applications, especially in genomics, distribution profiles  are expected to be quite similar across experimental units (e.g., subjects), and to vary only for a small fraction of the observations (e.g., genes). For the nDP, we have that  $Corr\left(G_j(A),G_{j'}(B)\right)= (1+\alpha)^{-1}>0$, where the expression does not depend on $\beta$: this is because the nDP assumes independence between atoms in separate distributions.
\subsection{Common Atoms Mixture Model}\label{Sec:CAMM}

The model defined through Equations \eqref{CAM1}--\eqref{CAM2} assumes a.s. discrete distributions. For modeling continuous distributions, one could follow established literature \citep{Ferguson83,Lo1984} and consider a nonparametric mixture model where \eqref{CAM1} is substituted by
\begin{equation}
\begin{aligned}
(y_{i_1,1},\ldots, y_{i_J, J})|	\, f_1,\ldots, f_J &\overset{ind.}{\sim} f_1\times\cdots\times f_J \qquad i_j=1,\ldots, n_j, \; j=1, \ldots , J\\
 f_j\left(\cdot\right)&=\int_{\Theta}\, p(\cdot|\theta)\,G_j(d\theta),\quad j=1, \ldots, J,
\end{aligned}
\label{CAMM1}
\end{equation}
where $p(\cdot|\theta)$ denotes an appropriate parametric continuous kernel density, and $G_j|Q\stackrel{i.i.d.}{\sim} Q$ as in \eqref{CAM2}.  In the rest of the paper, we will adopt Gaussian kernels, i.e. we assume $p\left(\cdot|\theta\right)$  to be Normal and $\theta=\left(\mu,\sigma^2\right)$ is a vector of location and scale parameters.

To simplify the computational algorithm,  we can introduce an alternative representation using two sequences of latent variables, $\bm{S}=\{S_j\}_{j\geq 1}$ and $\bm{M}=\{M_{i,j}\}_{i \geq 1,j\geq 1}$, describing -- respectively -- the clustering process at the distributional level 
and the observational level 
i.e. $S_j=k$ and $M_{i,j}=l$ if the observation $i$ in unit $j$ is assigned to the $l$-th observational cluster and the $k$-th distributional cluster. Thus we deal with the following model:
\begin{equation}
\begin{aligned}
y_{i,j}|\bm{ M, \theta} & \sim  N\left(\cdot | \theta_{M_{i,j}}\right), \quad
&M_{i,j}|\bm{S,\omega}  \sim \sum_{l=1}^{\infty} \omega_{l, S_j} \delta_l(\cdot),\\
\bm{\omega}_k | \bm{S}=\bm{\omega}_k &\sim GEM(\alpha), \quad
&S_{j}|\bm{\pi}  \sim \sum_{k=1}^{\infty} \pi_{k} \delta_k(\cdot),\\
\bm{\pi} &\sim GEM(\beta), \quad
&\theta_{l} \sim \pi(\theta_{l}),\:\: l\geq1,
\end{aligned}
\label{Membership}
\end{equation}
where we denoted with $\boldsymbol{\theta}=\{\theta_l\}_{l\geq 1}$. In the following, we consider a Normal-Inverse Gamma distribution for  $\theta_l= \left(\mu_l,\sigma^2_l\right)\sim NIG(m_0,\kappa_0,\alpha_0,\beta_0)$, i.e.  $\mu_{l}|\sigma^2_{l} \sim N\left(m_0, \sigma_{l}^2/\kappa_0\right)$ and $\sigma^2_{l} \sim IG\left(\alpha_0, \beta_0\right).$
 
\subsection{Common Atoms Model for Count Data}\label{Sec3:CAM_count}
In Section \ref{Sec:MicrobApp}, we consider an  application to microbiome data, which can be represented by abundance tables containing the observed frequency of a particular microbial sequence in a sample - or subject (unit). Here, we describe how the CAM  can be adapted to count data, characterized by skewness and zero-inflation typically observed in microbiome studies. Let $z_{i,j} \in \mathbb{N}$ be the observed count of  microbial sequence $i=1, \ldots, n_j$ in subject $j=1, \ldots, J$. Consequently, the vector $\bm{z}_j = \left( z_{1,j}, \ldots, z_{n_j,j}  \right)$ will denote the observed microbiome abundance vector of individual $j$. We embed model \eqref{CAM1}--\eqref{CAM2} in the rounded mixture of Gaussian framework of \citet{Canale2011}. See also \citet{Reich} and \citet{Canale2017a}, where the rounded mixture framework is compared to less flexible nonparametric mixtures of Poisson densities for count data. In order to define a probability mass function for the discrete measurements $z$, \citet{Canale2011}  consider a data augmentation framework by latent continuous variables $y$, such that
\begin{equation*}
f
\left(Z=j\right)=\int_{a_j}^{a_{j+1}}g\left(y\right)dy,  \:\:\: j\in\mathbb{N}
\end{equation*}
for a fixed sequence of thresholds $a_0<a_1<a_2<\ldots< a_{\infty}$ and for some density function $g(\cdot)$, such that $\int_{a_0}^{a_{+\infty}}g\left(y\right)dy=1$. Typically, the sequence of thresholds is set as  $\mathbf{a}=\{a_j\}_{j=0}^{+\infty}=\{-\infty,0,1,2,\ldots,+\infty\}$ and $g(\cdot)$ is a Dirichlet Process mixture density, to ensure a flexible representation of the table of counts. We propose a novel nested formulation, where $g(\cdot)$ is modeled as a CAM mixture \eqref{Membership}. More specifically, we consider
\begin{equation}
z_{i,j}|y_{i,j} \sim \sum_{g=0}^{+\infty}\delta_g(\cdot) \mathbf{1}_{ \left[a_g,a_{g+1}\right)}\left(y_{i,j}\right), \quad \quad
\label{newLik}
\end{equation}
where $y_{i,j}$ is distributed as in \eqref{Membership}. 
We will refer to this new setting as the Discrete Common Atoms Model (DCAM).

\section{Posterior Inference}
\label{SEC::4Post}
Typically, posterior samples for the nDP process have been obtained using a  truncated version of the Blocked-Gibbs Sampler \citep{Ishwaran2001a}, i.e. by choosing proper upper bounds for the infinite sums  that appear in \eqref{Membership}. Specifically, the model representation in \eqref{Membership} is useful to obtain  such an algorithm, which we detail in Section \ref{app:algo} of the Supplementary Material, where we also provide useful upper bounds to control the resulting truncation error. Here we present a novel nested version of the independent slice-efficient algorithm \citep{Walker2007, Kalli2011}. Compared to truncation-based algorithms, the proposed slice sampler has two main advantages: it allows to target  the true posterior distribution and it considerably decreases the computational time by stochastically truncating the model at the needed number of mixture components. 
The proposed slice sampling scheme can be easily extended to the nDP, and is related to the 
sampling scheme in \citet{Murray2013},  although their model is essentially different from ours.  In the following, we focus on the Common Atoms Mixture model \eqref{Membership}, as variations to accommodate for count data are straightforward. 

Let $p\left(y_{i,j}| \theta_l
 \right)$ denote a generic density function for the observation $y_{i,j}$, conditionally given $\theta_l$, let $\bm{\pi}=\{\pi_k\}_{k\geq 1}$ and $\bm{\omega}=\{\omega_{l,k}\}_{l,k \geq 1}$ 
be the two sets of weights, one referred to the distributional clusters, the other one referred to the observational clusters. Then, we can write:
\begin{equation*}
f\left(y_{i,j}|\bm{\theta},\bm{\omega},\bm{\pi}\right) = 
 \sum_{k\geq 1}\pi_k\, \sum_{l\geq 1} \,\omega_{l,k} \, p\left(y_{i,j}|\theta_l\right).
\end{equation*}

As in the classic slice sampler, we augment the model introducing two sets of latent variables controlling which components of the mixture are ``active'' and which can be ignored. More specifically, we introduce $\bm{u^D}=\{u^D_j\}_{j=1}^J$
-- where the D in the superscript indicates the distributional level -- and, within every unit $j=1,\ldots,J$, we define an inner sets of latent variables, 
$\bm{u_j^O}=\{u^O_{i,j}\}_{i=1}^{n_j}$, at the level of the observations. 
Moreover, we also consider the following deterministic sequences: $\bm{\xi}^D=\{\xi^D_k\}_{k\geq 1}$ and, for every $k$, $\bm{\xi_k}^O=\{\xi^O_{l,k}\}_{l\geq1}$. Then the model can be rewritten as
\begin{equation}
f_{\bm{\xi}^D,\bm{\xi}^O}\left(y_{i,j},u^D_j,u^O_{i,j}|\bm{\theta},\bm{\omega},\bm{\pi}\right) =  \sum_{k\geq 1}
\mathds{1}_{\{ u_j^D < \xi^D_k\}} \frac{\pi_k}{\xi^D_k}
\sum_{l\geq 1} \mathds{1}_{\{ u_{i,j}^O < \xi^O_{l,k}\}}\frac{\omega_{l,k}}{\xi^O_{l,k}}  p\left(y_{i,j}|\theta_l\right).
\end{equation}
Notice that if we assume $\xi^D_k=\pi_k$ and $\xi^O_{l,k}=\omega_{l,k}$, we recover the nested version of the efficient-dependent slice sampler. 
By introducing two sets of latent labels that identify the distributional ($\bm{S}$) and observational ($\bm{M}$) cluster in which the observation is allocated, we  get rid of the infinite sums in the previous equations. The complete likelihood for the entire dataset becomes 
\begin{equation}
\begin{split}
& f_{\bm{\xi}^D,\bm{\xi}^O}\left(\bm{y},\bm{u^D},\bm{u^O},\bm{M}, \bm{S}|\bm{\theta},\bm{\omega},\bm{\pi}\right)\\
& \qquad\qquad =  \prod_{j=1}^{J}
\mathds{1}_{\{ u_j^D < \xi^D_{S_j}\}} \frac{\pi_{S_j}}{\xi^D_{S_j}} \prod_{i=1}^{n_j}
\mathds{1}_{ \{ u_{i,j}^O < \xi^O_{ M_{i,j} ,S_j }\}}\frac{\omega_{ M_{i,j}, S_j }}{\xi^O_{ M_{i,j}, S_j }}   p\left(y_{i,j}|\theta_ {M_{i,j}}\right).
\end{split}
\label{LIKSLICE}
\end{equation}
Let $\phi(\cdot|\theta)$ and $\Phi(\cdot|\theta)$ denote the p.d.f. and the c.d.f. of a normal random variable with location-scale parameter $\theta$, respectively.
Then, if we assume $ p(y_{i,j}|\theta_{M_{i,j}})=\phi(y_{i,j}|\theta_{M_{i,j}})$ we recover the CAM model listed in \eqref{CAMM1}.
Alternatively, to recover the DCAM model for discrete data $\bm{z}$ as in \eqref{newLik}, it is sufficient to adopt the mixing kernel $p\left(z_{i,j}|\theta_{M_{i,j}}\right)=\Delta\Phi(a_{z_{i,j}};\theta_{M_{i,j}})=\Phi\left(a_{z+1}; \theta_{M_{i,j}}\right)-\Phi\left(a_{z}; \theta_{M_{i,j}}\right)$, obtained by integrating out the latent continuous variable. In a general framework, the nested slice sampler is obtained by looping over the full conditionals for $T$ iterations, according to the pseudo-code reported in Algorithm 1. For the DCAM, an additional step is added to update the latent continuous variable (see Step 1 of the algorithm in 
the Supplementary Material).
The computation of Steps 5, 6, and 7 is feasible, as we  stochastically truncate the number of mixture components to a sufficiently high integer to ensure that the two steps can be carried out exactly. Additional details for this procedure are reported in the Supplementary Material.

\begin{algorithm}[]
\SetAlgoLined
{
\For{$i=1,\ldots,T$}{
    1. Sample each $u_j^D$ from a uniform distribution $\mathcal{U}\left(0,\xi^D_{ S_j }\right)$.\\
    2. Sample each $u_{i,j}^O$ from a uniform distribution $\mathcal{U}\left(0,\xi^O_{ M_{i,j}, S_j }\right)$.\\
    3. Sample the proportions $\bm{v}$ for the SB weights independently from $v_k\sim Beta\left(a_k,b_k\right)$, where $a_k = 1+ \sum_{j=1}^{J}\mathds{1}_{\{S_j=k\}}$ and $b_k = \alpha+ \sum_{j=1}^{J}\mathds{1}_{\{S_j>k\}}$. This full conditional is obtained marginalizing $\bm{u^D}$ out.\\
    4. For each $k$, sample the proportions $\bm{u_k}$ independently from $u_{l,k}\sim Beta\left(a^k_l,b^k_l\right)$, where $a^k_l = 1+ \sum_{i=1}^{N}\mathds{1}_{\{M_{i,j}=l,S_j=k\}}$ and $b^k_l = \beta+ \sum_{i=1}^{N}\mathds{1}_{\{M_{i,j}>l,\:S_j=k\}}$. This full conditional is obtained collapsing both $\bm{u^D}$ and $\bm{u^O}$.\\
    5. Following \citet{Murray2013,Porteous2005}, we obtain more efficient updates trough partial collapsing, integrating over the inner level slice variables $\bm{u^O}$. Then, we sample from
    $$ \mathbb{P}\left(S_j=k|\cdots\right)\propto
    \mathds{1}_{\{ u_j^D < \xi^D_{k}\}} \frac{\pi_{k}}{\xi^D_{k}} \prod_{i=1}^{n_j}
    \omega_{ M_{i,j} , k }.$$\\
    
    
    
    6. Sample the observational labels from the following full conditional distribution:
    $$ \mathbb{P}\left(M_{i,j}=l|\cdots \right)\propto
    \mathds{1}_{ \{ u_{i,j}^O < \xi^O_{ l ,S_j }\}}\frac{\omega_{ l, S_j}}{\xi^O_{l ,S_j }}  p\left(y_{i,j}|\theta_l\right).$$\\
    7. Sample $\theta_l$ from a conjugate NIG.
 }
 \caption{Nested Slice-Efficient  Sampler for the Common Atoms Model}
}
\end{algorithm}
\vspace{.5cm}

\section{Analysis of microbial distributions of African Americans and rural Africans}\label{Sec:MicrobApp}
We apply the proposed modeling framework to the analysis of a microbiome dataset. Here, a primary goal is to study \emph{microbial diversity}, i.e. how the distribution of microbial units varies across subgroups of a population. 
Typically, summary statistics are used to capture characteristics of species' distributions, e.g. $\alpha$-diversity and $\beta$-diversity metrics such as Shannon's entropy and Bray-Curtis dissimilarity indexes, respectively  \citep{Whittaker2006}. However, those metrics do not fully capture the complexity of microbiome data, which poses distinctive statistical challenges \citep{Mao2017}. In particular, the data are recorded as counts of the observed microbial genome sequences. The resulting histograms are highly skewed and sparse, due to the many low- or zero- frequency counts and to the presence of a few dominant sequences (see Figure \ref{fig:histos}). Indeed, when compared across subjects, microbiota abundance data show a characteristic zero-inflation. 

The taxonomical classification of microbial species is typically conducted based on sequence alignments, e.g. through the use of 16S rRNA sequences:   ``practically identical'' sequenced tags ($\geq95\%$ of degree of similarity) are clustered together into the same \emph{phylotype}, and  referred to as an \emph{operational taxonomic unit} (OTU). Thus, for each specimen (e.g. fecal sample) obtained from a particular ecosystem (e.g. the gut), the number of recurrences of each OTU  is recorded \citep{Jovel2016, Kaul2017}. Collecting samples from distinct individuals leads to the construction of an \textit{abundance table}, a matrix formed by the OTU counts (taxa) observed in each sample. Let $\bm{Z}$ indicate a $n\times J$ abundance table where each entry $z_{i,j} \in \mathbb{N}$ is the frequency of the $i$-th OTU observed in the $j$-th subject, $i=1, \ldots, n$, $j=1, \ldots, J$, where $n$ represents the total number of OTUs. Thus, the vector $\bm{z}_j = \left( z_{1,j}, \ldots, z_{n, j}  \right)'$ denotes the observed microbiome sample  of individual $j$. 


To understand the varying composition of the microbiome in the population, we apply the  DCAM model proposed in Section \ref{Sec3:CAM_count} to the dataset from the study of \citet{OKeefe2015}, publicly available in the R package \texttt{microbiome}. The dataset contains the OTU counts of both healthy middle-aged African Americans (AA) and rural Africans (AF). The participants to the experiments were asked to follow their characteristic diet -- ``rural'' (low-fat and high-fiber) for AF and ``western'' (high fat and low-fiber) for AA -- for two weeks and then swap their diet regimes for other two weeks. During these two weeks, fecal samples were regularly collected to investigate the role of fat and fiber in the association between a specific diet and colon cancer risk. For our application, we focus on the abundance table obtained at the beginning of the experiment. Once we restrict our attention to the first time point, we find that 11 OTUs are absent across all the individuals. Therefore, they are removed from the dataset. However, since our model is designed to handle sparsity, we do not discard any underrepresented taxa, to avoid potential statistical power loss
\citep{McMurdie2014}. 
Our abundance table consists of 119 taxa measured for 38 patients. The heatmap of the data in log-scale, stratified by nationality, is shown in Figure \ref{fig:heatmap} in the Supplementary Material.\\
The varying sequencing depths also affect the so-called  \emph{library size}, i.e. the total frequencies  of the observed species (OTUs) in each subject sample. Let ${X}_j=\sum_{i=1}^{n}z_{i,j}$ indicate the library size for subject $j$ and let $\gamma_j=\bar{X}_j$ denote the corresponding average of the OTU frequencies. 
We incorporate the library sizes as a scaling factor in the latent level of the DCAM, i.e., 
\begin{equation}
y_{i,j}|\bm{ M, \mu,\sigma^2}  \sim  N\left(\gamma_j \cdot \mu_{M_{i,j}},\gamma_j^2\cdot \sigma^2_{M_{i,j}}\right) \iff
\frac{y_{i,j}}{\gamma_j}|\bm{ M, \mu,\sigma^2}  \sim  N\left(  \mu_{M_{i,j}}, \sigma^2_{M_{i,j}}\right).
\label{newLikLS}
\end{equation}
Both the mean and the variance of the latent continuous random variable are decomposed multiplicatively into the deterministic term $\gamma_j$ that describes the depth of the sequencing, and two stochastic terms that capture the intensity $\mu_{M_{i,j}}$ and the uncertainty $\sigma^2_{M_{i,j}}$ behind the OTU counts, respectively. \\
We adopt standard prior settings for  all the hyperparameters $\left(m,\kappa,\alpha_0,\beta_0,a_\alpha,b_\alpha,a_\beta,b_\beta\right)$. Following an empirical Bayes rationale, we set $m$ and $\kappa$ to be equal to the grand mean and the inverse of the overall sample variance. According to \citet{Rodriguez2008}, we then set $\beta_0=1$ and $\alpha_0=a_\alpha=b_\alpha=a_\beta=b_\beta=3$. A MCMC sample of 100,000 iterations was collected after a burn in period of the same length. Convergence of the MCMC was assessed based on visual inspection and standard convergence diagnostics \citep{oro22547}.

\begin{figure}[!ht]
  \begin{minipage}[b]{0.5\linewidth}
    \centering
	\includegraphics[scale=.53]{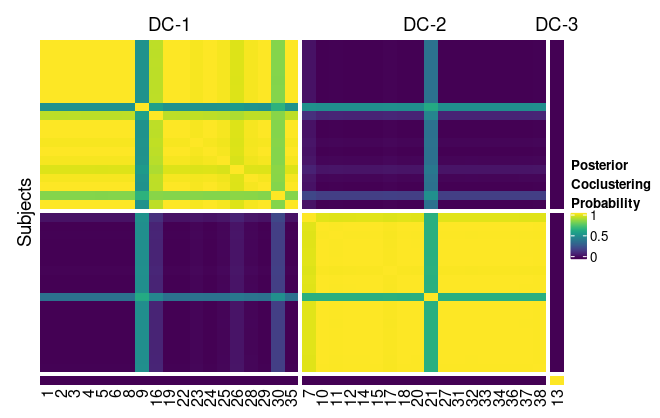}
  \end{minipage}%
  \begin{minipage}[b]{0.5\linewidth}
    \centering
\begin{tabular}[b]{lccc }
 \toprule
Cluster & DC-1 & DC-2 &  DC-4 \\ \midrule
  Cardinality & 19 & 18 & 1\\
\midrule
  Africans & 2& 14 &  1\\ 
  Americans & 17 & 4 &  0\\
\midrule
  Female & 11 & 6 &  0\\
  Male   &  8 & 12 &  1\\
  \bottomrule
\end{tabular}
\end{minipage}
	\centering 
	\caption{Left: pairwise posterior probability matrix of coclustering among the 38 subjects. A partition of the subjects' distributions into three clusters is obtained after minimization of the posterior expected Variation of Information loss function. Right: Table reporting the clusters' characteristics.}
	\label{CCMicrob}
\end{figure}
\textbf{Distributional cluster analysis.} To obtain an estimate for the distributional clustering, we first compute the posterior pairwise co-clustering matrix.  
From this matrix, we estimate the optimal partition by considering a decision-theoretic approach and minimizing the expected posterior loss under a specific loss function. 
We follow \citet{Wade2015}, who propose to rely on the minimization of the Variation of Information loss function developed by \citet{MeilaM2007}. 
The results are reported in Figure \ref{CCMicrob}, where we also summarize the main characteristics of these distributional clusters (DCs) in terms of cardinality, nationality, and gender. It is remarkable how the different subpopulations of microbiome populations are captured by our model: in fact, Cluster DC-1 contains almost all the AA subjects, while Cluster DC-2 is composed mostly of AF. Cluster DC-3 contains only one subject, whose microbiome distribution is substantially unique. The resulting DCs capture relevant distributional characteristics and the diversity of the microbiomes. In particular, the Shannon index  \citep{Shannon48} or the Simpson index are often used to measure the $\alpha$-diversity of a microbiome community, i.e. the richness (number) and evenness (frequencies' similarity) of the different OTUs observed in a sample. 
Conditionally on the optimal configuration, we compute 9 summary statistics for each subject. The DCs capture the different levels of $\alpha$-diversity of the microbiome subpopulations. Indeed, the Shannon index and the Simpson Index vary substantially across the groups. In detail, the distributional cluster DC-1 is characterized by microbiome distributions with shorter range, lower standard deviations, skewness, and kurtosis than DC-2. However, DC-2 also show less richness/diversity than DC-1. See Figure \ref{Boxplots} in the Supplementary Material. Therefore, we expect that the microbiomes clustered in DC-2 are more likely to contain a small fraction of highly prominent OTUs.  
To confirm this intuition, let $z_{(i),j}$ represent the $i$-th most frequently observed OTU  in subject $j$. We define the cumulative relative frequency (CRF) for subject $j$ as $CRF_j(i) = \sum_{l=1}^i z_{(l),j}/\sum_{i=1}^{n} z_{i,j}$. Figure \ref{ECDF} shows the CRFs for all the subjects colored by the DCs. The CRF curves in DC-2 tend to get very close to 1 within the first 25 most abundant OTUs, showing that the relative frequencies are dominated by few, but highly expressed taxa. At the same time, the CRF curves in DC-1 increase with a slower pace, meaning more heterogeneity in the microbiome subpopulations. The CRF curve of the single subject in DC-3 increases much more slowly, indicating a peculiar microbiome, richer and more diverse than any other. 
\begin{figure}[t!]
	\centering 
	\includegraphics[scale=.6]{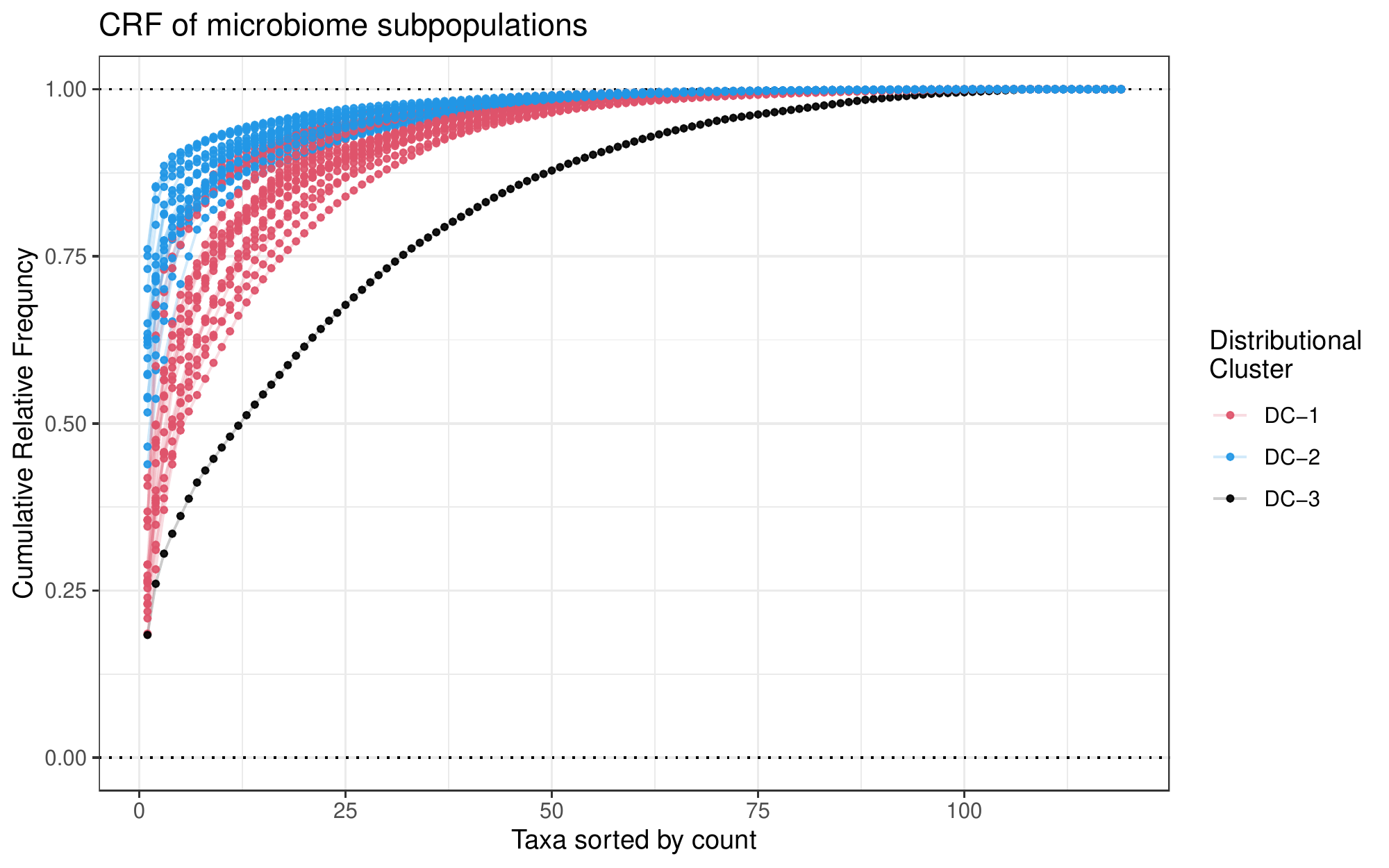}
	\caption{Cumulative Relative Frequency of the OTU abundances, sorted by decreasing order. Each color represents a DC. The lower the line, the richer and more diverse is the microbiome.}
	\label{ECDF}
\end{figure}	
We compute the median abundance of each OTU stratified by DC. In both cluster DC-2 and cluster DC-3, the leading OTU is the \emph{Prevotella melaninogenica}. On average, it represents 60\% of the observed counts in each individual in DC-2 and the 18\% in DC-3. Cluster DC-1 is more diverse: the two most expressed OTUs are the \emph{Bacteriodes vulgatus} and the \emph{Oscillospira guillermondii} that on average represent the 15\% and the 12\% of the subjects' library size, respectively. Cluster DC-3 is also characterized by a high proportion of \emph{Faecalibacterium prausnitzii} (7\%). \\

\textbf{Observational cluster analysis.} We further investigate the observational clusters (OC) induced by DCAM.
Minimizing the Variation of Information we find 8 OCs, representing different intensities of the latent process underlying the counts. 
For a visual comparison, we report in Figure \ref{Boxplots_OC} the boxplots of the taxa counts grouped by OC, with the value of the median superimposed. For simplicity, we group the 8 OCs in three macro clusters representing the \emph{abundance classes} ({Low}, {Medium}, and  {High}). Heatmaps showing the prevalence of each OTU in every abundance class are reported in the Supplementary Material.
\begin{figure}[ht!]
	\centering 
	\includegraphics[scale=.5]{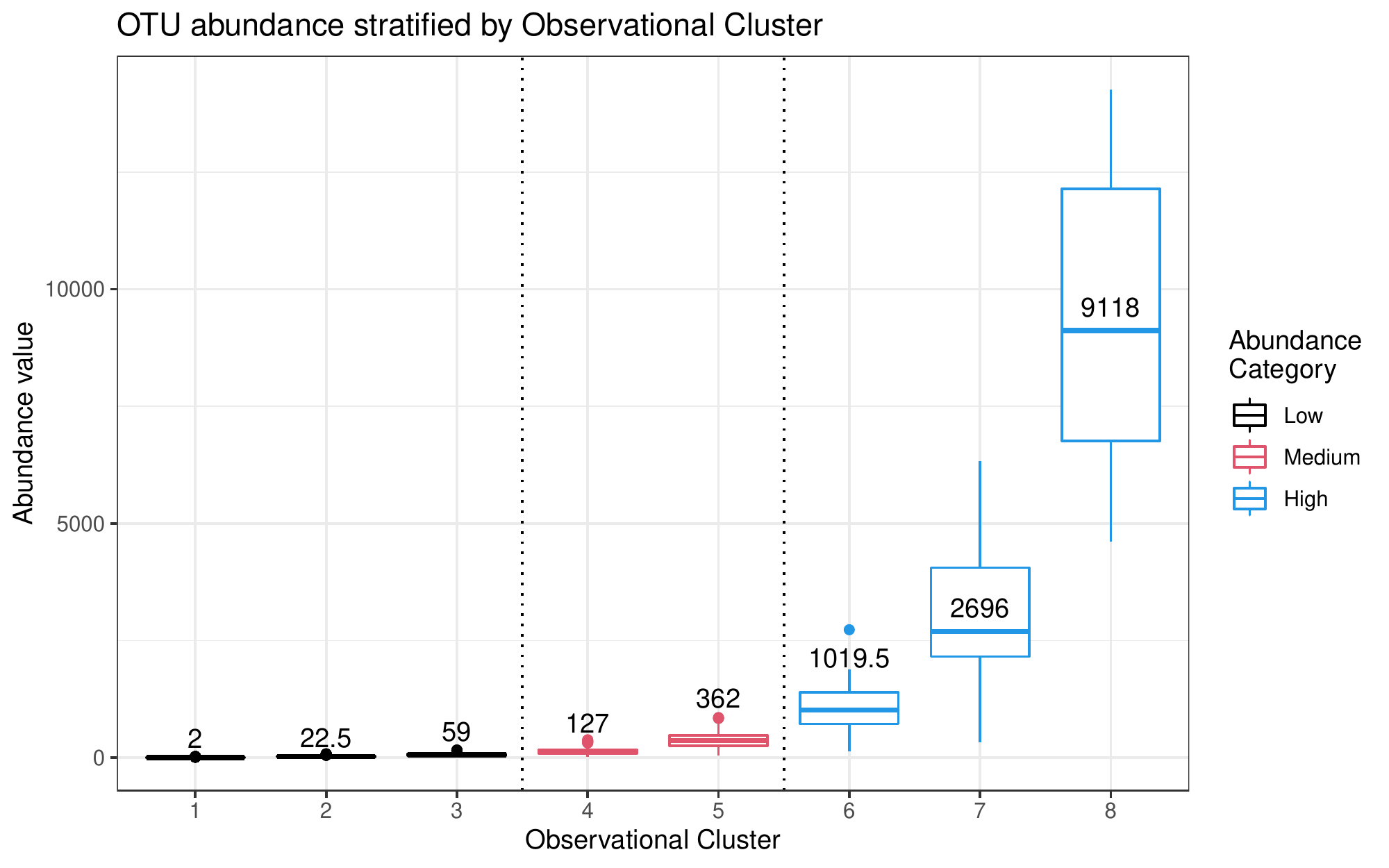}
	\caption{Boxplots of microbiome abundance counts stratified by observational clusters. We can recover three macro-clusters, with Low, Medium and High level of expression. The count median of each category is superimposed.}
	\label{Boxplots_OC}
\end{figure}	

Finally, the distributional and observational results can be combined to discover more informative patterns, relating OTUs and subjects. Here, we investigate the co-expression structure among the most expressed OTUs in DC-1 and DC-2. To do so, we first stratify the subjects by distributional clusters (DC-1 and DC-2) and remove the OTUs that, across all individuals, are always assigned to the Low abundance class. With the remaining 12 OTU, we compute two pairwise co-occurrence matrices ($PCM_k$) as
$PCM_k(l,g) = \sum_ {h=1}^{n_k} \mathds{1}_{\{AC(g)=AC(l)\}} /n_k,$
i.e. the percentage of times that OTU $l$ and OTU $g$ have been assigned to the same abundance class (AC) 
across the $n_k$ individuals assigned to DC $k=1,2$.
We plot two co-occurrence networks among the selected OTUs in Figure \ref{fig:my_network}. Taxa $l$ and $g$ are linked if $PCM_k(l,g)=PCM_k(g,l)>0.5$. The nodes are colored according to the modal abundance class. Again, the \emph{Prevotella malaninogenica} and the \emph{Prevotella oralis} are both highly expressed and co-occurent in DC-2, while in DC-1 they fall in the {Low} abundance class and are not linked. In DC-1, highly and co-occurent taxa are the \emph{Bacteriodes vulgatus}. These results are in line with  well-established results in the literature, since subjects with a preponderance of \emph{Prevotella spp.} are more likely to consume fibers, while diets richer protein and fat diet - typical of western diets - lead to a predominance of \emph{Bacteroides spp.} \citep{Graf2015,Preda2019}.
\begin{figure}[ht!]
    \centering
    \includegraphics[scale=.4]{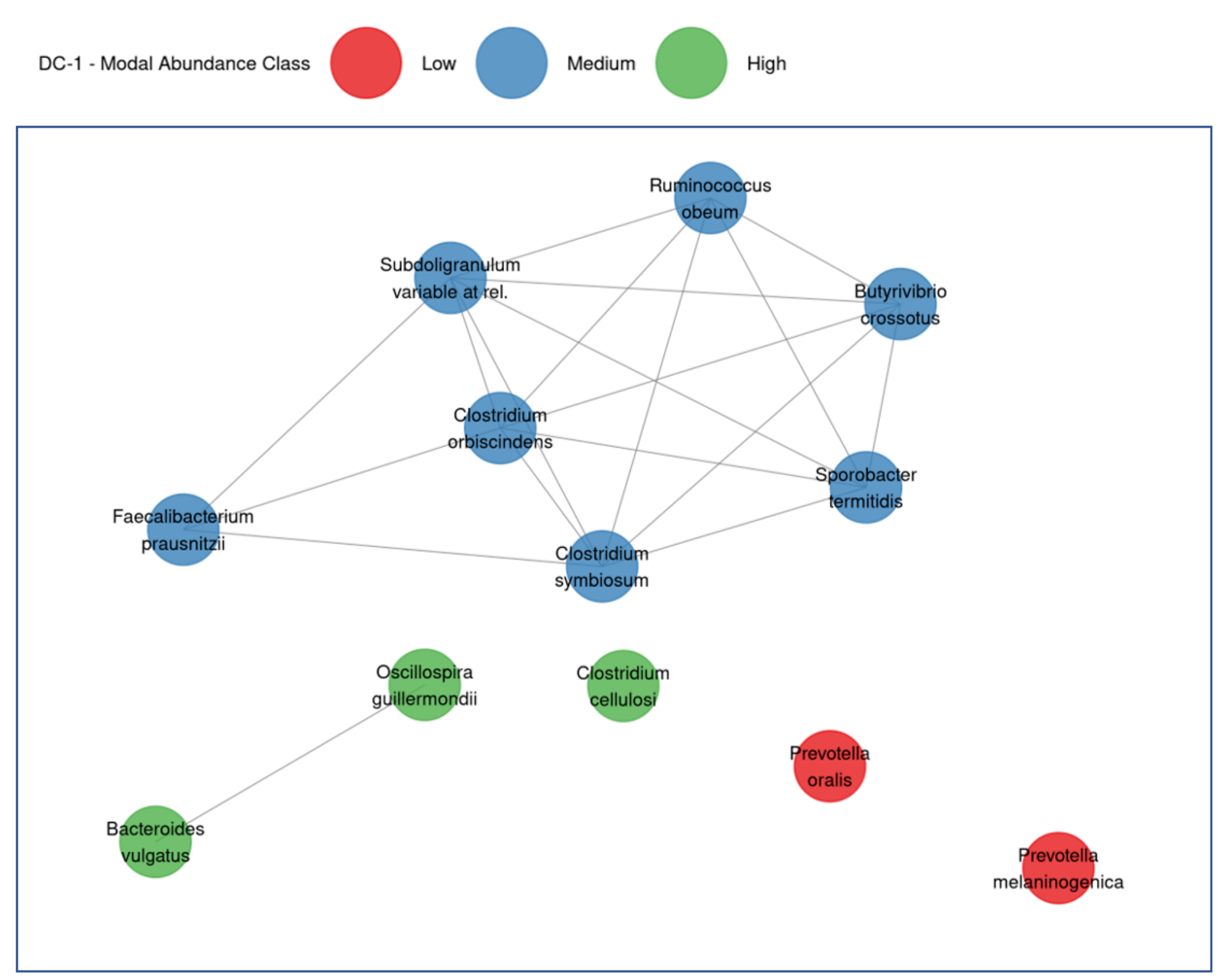}
    \includegraphics[scale=.4]{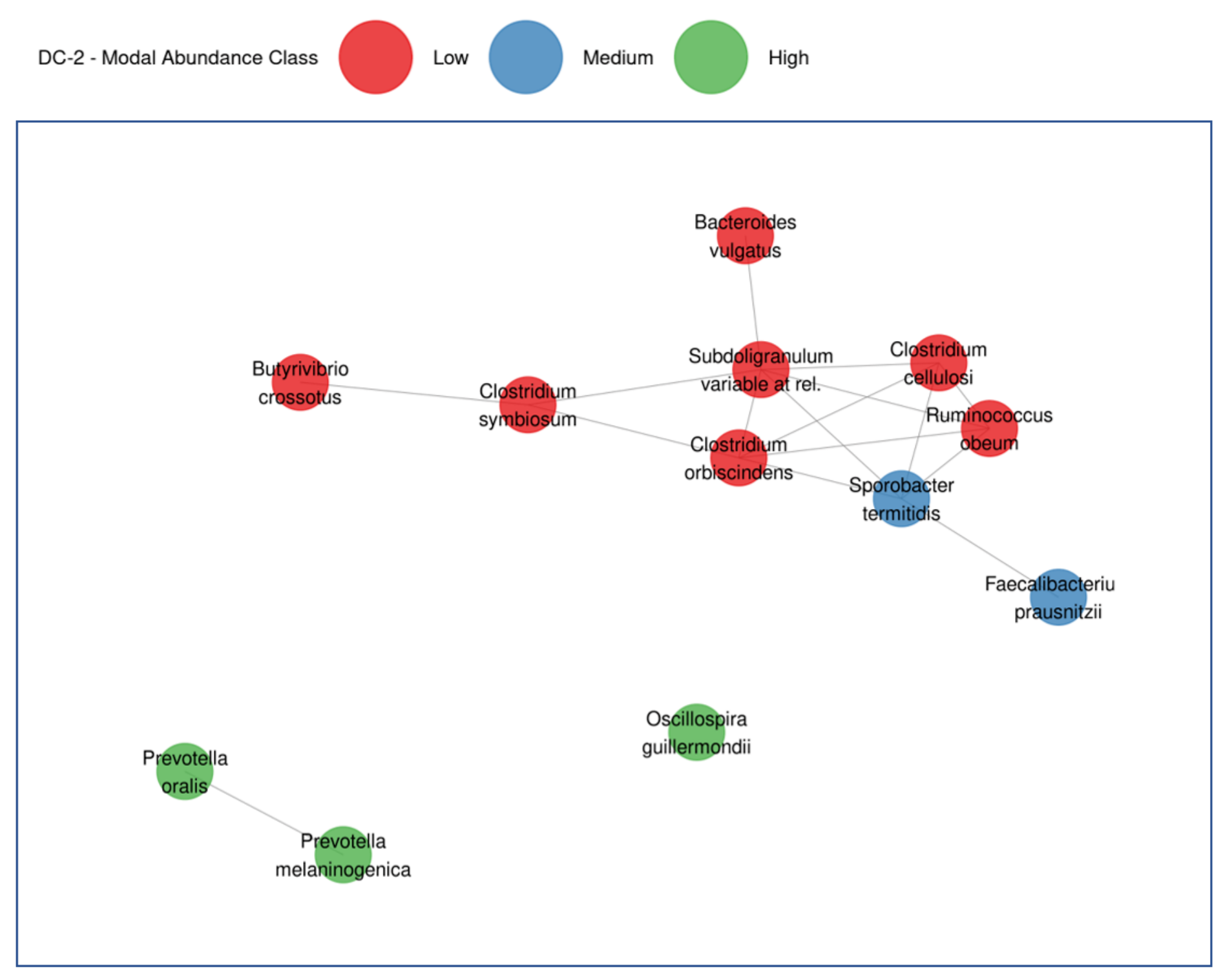}
    \caption{Co-expression networks among OTUs reporting a subset of most expressed microbes for DC-1 (left panel) and DC-2 (right panel).}
    \label{fig:my_network}
\end{figure}

\section{Simulation study}
\label{SEC::6SimStudy}
We test the  performances of the proposed methodology for continuous (CAM) and discrete measurements (DCAM) within a simulation study comprised of three scenarios. For every scenario, we generate the units containing the observations from highly overlapping mixture densities. We want to assess our model's ability to recover the ground truth by recognizing the units sampled from the same mixture density (i.e. identify the distributional clusters -- DC) and the observations generated from the same mixture component (i.e. identify the observational clusters -- OC), for increasing number of observations in each unit, $n_j$, or  for increasing number of units, $J$. We adopt the same prior specification as in the case study and estimate the best partitions by minimizing the Variation of Information given the MCMC output. We now describe the three scenarios:\\
\emph{\textbf{Scenario 1 - CAM}}. We define 6 different distributions of the simulated data $\bm{Y}_{h}$, as
$$  \bm{Y}_{h} \sim \sum_{g=1}^{h} \frac{1}{g}N(m_g,0.6), \text{ where } m_g\in\{0,5,10,13,16,20\} \text{ and }  h=1,\ldots,6 .$$
    From each of these distributions, we sample two units,  therefore $J=12$. The true number of DCs and OCs is 6 in both cases. To assess how the model behaves with asymmetries in the units' sample sizes, we follow two different approaches. \textbf{Case A}: all the units have the same cardinality $n_j=n_A$, where $n_A\in \{25,50,75\}$. \textbf{Case B}: each unit has cardinality $n_j$ proportional to the number of mixture components it contains. Specifically, $n_{j}=
    n_B \cdot j$ for $j=1,\ldots,6$ and $n_B \in \{5,10,20\}$. \\
\emph{\textbf{Scenario 2 - CAM}}.  
Four highly overlapping mixtures are considered:
    \begin{align*}
    \bm{Y}_{1} &\sim 0.75N(0,0.6)+  0.25 N(3,0.6),\quad 
    \bm{Y}_{2} \sim 0.25N(0,0.6)+  0.75 N(3,0.6),\\
    \bm{Y}_{3} &\sim 0.33N(0,0.6)+  0.34 N(-2,0.6)+  0.33 N(2,0.6),\\
    \bm{Y}_{4} &\sim 0.25N(0,0.6)+  0.25 N(-2,0.6)+  0.25 N(2,0.6)+  0.25 N(10,1).
    \end{align*}
The true number of DCs is 4 and there are 5 OCs, corresponding to the 5 different normal distributions that constitute the mixtures. We keep the number of observation per unit constant, equal to $n_j=40$ for any $j$. Instead, we vary the number of units sampled from each distribution, indicated as $r=1,\ldots,6$. Therefore, $J_r=4\cdot r$, i.e. the total number of considered units ranges from $J_1=4$ to $J_6=24$. In this way, we can investigate the estimated DC structures as the total number of units increases.\\
 \emph{\textbf{Scenario 3 - DCAM}}. First, let $\delta_x$ denote a point mass placed on point $x$ and let $\mathcal{U}_d\left(q,Q\right)$ represent a Uniformly Discrete distribution over the set of integers $\{q,\ldots,Q\}\subset \mathbb{Z}$. We consider three different discrete mixtures, from which we sample $J=10$ units: 
    \begin{align*}
    \bm{Y}_{g} &\sim \sum_{b=1}^2\omega_b \delta_{b-1} + \omega_3 \:\mathcal{U}_d\left(0,Q_g\right) \text{ with } g=1,2,3 \text{ and } Q_g\in\{10,50,100\},
    \end{align*}
    with $\omega_g=n_g/\sum_{l=1}^3 n_l\: ,\: g=1,\ldots,3$ denoting the mixture weights.  We set $\omega_1=\omega_2$ by generating the $n_1=50$ observations equal to zero and $n_2=50$ equal to one to simulate a case of low value inflation. We investigate the performance of the model in 6 cases, distinguished by the number of observations assigned to the third mixture component, i.e. $n_3\in\{10, 15 ,25 , 50,75,100\}$. 
    We design this simulation study to test how DCAM perform on distributions that are similar to typical microbiome samples, raising the same type of challenges. The number of true DCs is fixed equal to 3. However, there is no clear number of true OCs in this case. To assess the grouping at the level of the observations, we assume the following sets as ground truth, mimicking the segmentation in abundance levels of Section \ref{Sec:MicrobApp}. We postulate 4 OCs, where the first set contains ``low-expressed'' observations (i.e. constituted of zeros and ones). The remaining 3 groups are obtained partitioning the support into abundance classes corresponding to the intervals $\left[2,10\right]$,$\left(11,50\right]$ and $\left(51,100\right]$.\\
We report an illustrtion of the mixtures distributions of each scenario in the Supplementary Material.\\
For each scenario, we also run a nDP mixture model 
for the case with the highest number of observations.
In Table \ref{TABRAND} we assess the goodness of the estimated optimal partition by comparing the number of detected clusters, computing the Adjusted Rand Index \citep[ARI -][]{Hubert1985} between the estimated optimal partition and the ground truth. Moreover, we report the normalized Frobenius distance \citep{Horn2013} between the estimated posterior pairwise coclustering matrices and the true coclustering structures, defined as follows. Given two $p\times p$ matrices $A=\{a_{ij}\}_{i,j=1}^p$ and $B=\{b_{ij}\}_{i,j=1}^p$, we define $NFD(A,B)=\sum_{i,j=1}^p (a_{ij}-b_{ij})^2/p^2$. 
From Table \ref{TABRAND}, we can appreciate how the model can recover the ground truth, even for small sample sizes. In particular, the NFD between the distributional clustering structures approaches zero as the sample size increase. The same holds for the ARI index, that shows how the truth is recovered by the estimated best partition.
We see how CAM misassigned a few observations in the wrong OCs in Scenario 2. This is due to the fact that the different mixture components are highly overlapping. Nevertheless, CAM and DCAM perform really well in Scenarios 1 and 3, respectively, where the true OC are well separated. 
Lastly, it is evident how the overlap of the data impacts the estimated partitions of the nDP, both at the distributional and at the observational level. In particular, when highly overlapping discrete data are considered (Scenario 3), it collapses all the units in a single DC.

\begin{table}
\begin{center}
\begin{tabular}{c c c c c c c c}
\toprule
 \textbf{Scenario 1} & $n_A=25$ & $n_A=50$ & $n_A=75$ & $n_B=5$ & $n_B=10$ & $n_B=20$ & nDP\\
\midrule
  DC-D/T & 4/6 & 6/6 & 6/6 & 4/6 & 5/6 & 6/6 & 5/6\\
  DC-ARI & 0.421 & 1.000 & 1.000  & 0.542 & 0.718 & 1.000  & 0.718\\
  DC-NFD & 0.123 & 0.007 & 0.004  & 0.094 & 0.058 & 0.002  & 0.056\\  
  OC-D/T & 4/6 & 6/6 & 6/6 & 5/6 & 6/6 & 6/6 & 6/6\\
  OC-ARI &  0.925 & 0.988 & 0.973 & 0.964 & 0.970 & 0.964 & 0.353        \\
  OC-NFD &  0.082 & 0.102 & 0.115 & 0.041 & 0.064 & 0.098   & 0.134    \\  
\midrule
  \textbf{Scenario 2} & $J_1=4$ & $J_3=8$ & $J_3=12$ & $J_4=16$ & $J_5=20$ & $J_6=24$ & nDP\\
\midrule
 DC-D/T & 3/4 & 4/4 & 5/4 & 5/4 & 4/4 & 4/4 & 4/4\\
 DC-ARI & 0.000 & 1.000 & 0.891 & 0.918 & 1.000 & 1.000 & 1.000\\
 DC-NFD & 0.081 & 0.003 & 0.022 & 0.019 & 0.003 & 0.008& 0.001 \\  
 OC-D/T & 5/5 & 5/5 & 5/5 & 4/5 & 5/5 & 5/5 & 2/5\\
 OC-ARI & 0.665 & 0.756 & 0.714 & 0.629 & 0.758 & 0.768& 0.092\\
 OC-NFD & 0.113 & 0.124 & 0.143 & 0.152 & 0.129 & 0.143 & 0.149 \\
 \midrule

  \textbf{Scenario 3} & $n_3=10$ & $n_3=15$ & $n_3=20$ & $n_3=50$ & $n_3=75$ & $n_3=100$ &nDP \\
 \midrule
 DC-D/T & 7/3 & 2/3 & 3/3 & 5/3 & 4/3&3/3 & 1/3  \\
 DC-ARI &  0.115 & 0.366 & 1.000 & 0.695 & 0.826 & 1.000& 0.000 \\
 DC-NFD & 0.259 &0.251 &0.035 &0.076 &0.057 & 0.000& 0.640\\  
OC-D/T & 4/4 & 4/4 & 5/4 & 3/4 & 4/4 & 6/4 & 10/4  \\
OC-ARI & 0.999 & 0.945 & 0.966 & 0.973 & 0.953 & 0.937&0.534 \\
OC-NFD & 0.722 & 0.238 & 0.414 & 0.454 & 0.338 & 0.151 & 0.740\\  
 
\bottomrule
\end{tabular}
\caption{Distributional (DC-) and observational (OC-) clustering performance for CAM, DCAM and nDP evaluated according the number of detected clusters over the truth (D/T), the Adjusted Rand Index (ARI) and the normalized Frobenius distance (NFD) between posterior pairwise coclustering matrices.}
\label{TABRAND}
\end{center}
\end{table}

\section{Discussion}
\label{SEC::7Conclusion}

We have introduced a nested nonparametric model that allows investigating distributional heterogeneity in nested data. The proposed Common Atoms Model allows a two-layered clustering at the distributional and observational level, similarly to the nDP of \citet{Rodriguez2008}. By construction, our model formulation allows the sharing of atoms with different weights across distributions, and it does not suffer from the degeneracy properties  that occurs in the nDP, as noted by \citet{Camerlenghi2018} whenever there is a tie between atoms. The Common Atoms Model specification is appealing and convenient for a variety of reasons: it is simple, allows a more refined description of distributional clusters, and it is computationally efficient thanks to the implementation of a nested version of the independent slice-efficient sampler.  We have extended the methodology to take into account the modeling and clustering of discrete distributions, by considering a rounded mixture of Gaussian kernels as in \citet{Canale2011}. 
We applied our methodology to a real microbiome dataset, aiming to cluster individuals characterized by similar taxa distributions. Controlling for each subject's library size, we grouped the data minimizing the Variation of Information loss function, and showed how the model detects clusters catching main differences among the distributions. In our application, the distributional clustering we recover distinguishes among dietary patterns, discriminating African high fiber from Western high fats diets. The observational clustering provides insights about the abundance levels among taxa and helps the identification of co-expression networks.     We also assess the performance of our modeling approach through a simulation study where the data are 
simulated from highly overlapping distributions.

The application of the proposed model to the real data set is limited by the type and number of clinical and demographic covariates that are available. 
If additional covariates were available, they could be used to define more complex dependencies, e.g. by constructing dependent random measures with covariate-dependent weights as in \citet{MacEachern2000} \citep[see, also][]{barrientos2012} or to build risk-prediction models. Another interesting extension considers the incorporation of a time dimension and the study of 
how distributional clusters vary across time.
We leave these directions to future investigation. The code employed for this paper is openly available at \url{https://github.com/Fradenti/CommonAtomModel}

\clearpage

\section*{Supplementary Material}
\appendix

\section{Proofs} \label{Proof}
\subsection{Proof of Equation \eqref{coclustG}}  \label{app:eq_cluster}
Let $G_j$ and $G_{j'}$, with $j \not =j'$, be two random probability measures as defined in \eqref{CAM1}-\eqref{CAM2}. Then,
\begin{align*}
\mathbb{P}\left(G_{j}=G_{j'}|Q\right) &= \sum_{k \geq 1}\mathbb{P}\left(G_j=G_{j'}=G^{*}_{k}|Q\right)=
\sum_{k \geq 1}\mathbb{P}\left(G_j=G^{*}_{k} , G_{j'}=G^{*}_{k}|Q\right)\\
&= 	\sum_{k \geq 1}\mathbb{P}\left(G_j=G^{*}_{k}|Q\right)\mathbb{P}\left( G_{j'}=G^{*}_{k}|Q\right)=\sum_{k \geq 1} \pi_k^2>0.
\end{align*}
As a consequence we get
\begin{equation*}
\mathbb{P}\left(G_j=G_{j'}\right) = \mathbb{E}\left[\mathbb{P}\left(G_j=G_{j'}|Q\right) \right] = \mathbb{E}\left[\sum_{k \geq 1} \pi_k^2 \right]= \sum_{k \geq 1} \mathbb{E}\left[\pi_k^2 \right],
\end{equation*}
exploiting the stick--breaking representation of the $\pi_k$'s we have
\begin{align*}
  \mathbb{P}\left(G_j=G_{j'}\right) & =    \sum_{k \geq 1} \E{ V_k^2 \prod_{i=1}^{k-1} (1-V_i)^2} = \sum_{k \geq 1} \frac{B(3, \alpha )}{B(1, \alpha)}  \left[ \frac{B(1,\alpha+2)}{B(1, \alpha)}   \right]^{k-1}
\end{align*}
where we denoted by $B(x,y) = \Gamma (x) \Gamma(y) / \Gamma (x+y)$ the beta function. Some simple calculations show that
\begin{align*}
    \mathbb{P}\left(G_j=G_{j'}\right) & = 
    \frac{2}{(1+\alpha)(2+\alpha)} \sum_{k \geq 0} \left[\frac{\alpha}{\alpha+2}\right]^{k} =
    \frac{1}{\alpha+1},
\end{align*}
and then \eqref{coclustG} follows.

\subsection{Proof of Equation \eqref{xijxij}}
Let $y_{i,j}| G_j\sim G_j$ and $y_{i',j'}| G_{j'}\sim G_{j'}$ be two observations coming from two probability measures both sampled from $Q$ for $j \not = j'$. Then, 
\begin{align}
\Prob{y_{i,j} =y_{i',j'} }&=\E{\Prob{y_{i,j}=y_{i',j'}|G_j,G_{j'}}}  \nonumber \\ &= \E{  \frac{1}{1+\alpha}\Prob{y_{i,j}=y_{i',j'}|G_j=G_{j'}}+
\frac{\alpha}{1+\alpha}\Prob{y_{i,j}=y_{i',j'}|G_j\neq G_{j'}}} \nonumber\\&=  \frac{1}{1+\alpha}\E{\sum_{r\geq 1} \omega^2_{r,j}}+
\frac{\alpha}{1+\alpha}\E{\sum_{r\geq 1} \omega_{r,j}\omega_{r,j'} }.
\label{eq:xx_E}
\end{align}
To conclude the proof we evaluate the two expected values in the last equation. As for the first one, it is easy to observe that 
\[
\E{\sum_{r\geq 1} \omega^2_{r,j}} = \frac{1}{1+\beta}
\]
using the stick-breaking representation of the weights as in Section \ref{app:eq_cluster}. As for the second expected value we exploit the independence across the $\omega_{r,j}$'s, for different values of $j$, and we get
\begin{align*}
 \E{\sum_{r\geq 1} \omega_{r,j}\omega_{r,j'} } &=  \sum_{r\geq 1} \E{\omega_{r,j}}\E{\omega_{r,j'} } = \sum_{r \geq 1} \E{\omega_{r,j}}^2\\
 & = \sum_{r \geq 1} \left[ \E{V_r} \prod_{i=1}^{r-1} \E{1-V_i} \right]^2
 = \sum_{r \geq 1} \left[ \frac{1}{1+\beta} \left( \frac{\beta}{1+\beta} \right)^{r-1} \right]^2\\
 & = \frac{1}{2\beta+1}
\end{align*}
where the last equality follows by straightforward calculations.
Substituting the previous expressions in \eqref{eq:xx_E} we finally obtain
\begin{equation*}
   \Prob{y_{i,j} =y_{i',j'} }= \frac{1}{1+\alpha}\frac{1}{1+\beta} + \frac{\alpha}{1+\alpha}\frac{1}{2\beta+1}
\end{equation*}
and Equation \eqref{xijxij} follows.

\subsection{Proof of Equations \eqref{eq:cov}--\eqref{eq:corr_same}}
Suppose that the  $G_j$'s are defined on a Polish space $\left(\mathbb{X},\mathcal{X}\right)$ and consider $A, B \in \mathcal{X}$.
Recall that $G_{j},G_{j'}|Q \stackrel{i.i.d.}{\sim} Q$, where
$Q= \sum_{k \geq 1} \pi_k \delta_{G_k^*}$. In the following, for the sake of notational simplicity and without loss of generality, we suppose that $j=1$ and $j'=2$.  
We now focus on the proof of \eqref{eq:cov}, for this reason we first evaluate
\begin{align*}
	&\E{G_1(A)  G_2(B)}=	\E{\E{G_1(A) \cdot G_2(B)|Q}}\\
& \qquad	= \E{ \sum_{k \geq 1 } \pi_k^2 G_k^*(A) G_k^* (B)
+\sum_{k_1 \not = k_2} \pi_{k_1} \pi_{k_2} G_{k_1}^*(A) G_{k_2}^*(B)}
\end{align*}
Since the $G_k^*$'s are independent and identically distributed and thanks to the fact that 
\[
\Prob{G_1=G_2} = \E{\sum_{k \geq 1 } \pi_k^2},
\]
we can equivalently write
\begin{equation*} 
    	\E{G_1(A) G_2(B)}= 
    	\Prob{G_1=G_2} \E{G_1^*(A) G_1^*(B)} + 
    	\Prob{G_1 \not = G_2} \E{G_1^*(A) G_2^*(B)}.
\end{equation*}
In view of Equation \eqref{coclustG}, the previous expression boils down to the following one
\begin{equation} \label{eq:EG1G2}
\E{G_1(A)  G_2(B)}=
 \frac{1}{\alpha+1} \E{G^*_1(A) G^*_1(B)}+\frac{\alpha}{\alpha+1} \E{G^*_{1}(A) G^*_{2}(B)}.
\end{equation}
We now focus on the evaluation of the two expected values in \eqref{eq:EG1G2}.
The first one can be expressed as
\begin{align*}
	\E{G^*_{1}(A)G^*_{1}(B)} =& \E{ \sum_{l\geq 1}\omega_{l,1} \delta_{\theta_l}(A) \cdot \sum_{l\geq 1}\omega_{l,1} \delta_{\theta_l}(B)}\\
	=&\E{ \sum_{l\geq 1}\omega^2_{l,1} \delta_{\theta_l}(A\cap B)} +\E{  \sum_{l\geq 1}\sum_{r\neq l} \omega_{l,1}\omega_{r,1} \delta_{\theta_l}(A)\delta_{\theta_r}(B)}\\
	=&\E{ \sum_{l\geq 1}\omega^2_{l,1}}H(A\cap B) +\left(1-  \sum_{l\geq 1}  \E{ \omega^2_{l,1} }\right)H(A)H(B)\\
	=& \frac{1}{\beta+1}  H(A\cap B) +\frac{\beta}{\beta+1}H(A)H(B),
\end{align*}
where we used the fact that $H$ is the distribution of the atoms and 
\[
\E{ \sum_{l\geq 1}\omega^2_{l,1}} = \frac{1}{1+\beta}.
\]
The second expectation in \eqref{eq:EG1G2} can be evaluated as follows:
\begin{align*}
\E{G^*_{1}(A) G^*_{2}(B)} =& \E{ \sum_{r\geq 1}\omega_{r,{1}} \delta_{\theta_r}(A) \cdot \sum_{l\geq 1}\omega_{l,{2}} \delta_{\theta_l}(B)}\\
=&\E{ \sum_{r\geq 1}\omega_{r,{1}}\omega_{r,{2}} \delta_{\theta_r}(A\cap B)} +\E{  \sum_{r\neq l} \omega_{r,{1}}\omega_{l,{2}} \delta_{\theta_r}(A)\delta_{\theta_l}(B)}\\
=&\E{ \sum_{r\geq 1}\omega_{r,{1}}\omega_{r,{2}}}
H(A\cap B) + \E{  \sum_{r\neq l} \omega_{r,{1}}\omega_{l,{2}} }H(A)H(B).
\end{align*}
We note that the previous equality holds true in particular when $A=B=\mathbb{X}$. In that case
\[
1= \E{G^*_{1}(\mathbb{X})\cdot G^*_{2}(\mathbb{X})} = \sum_{r\geq 1} \E{\omega_{r,{1}}}\E{\omega_{r,{2}}}
H(\mathbb{X}) +   \sum_{r\neq l} \E{\omega_{r,{1}}}\E{\omega_{l,{2}} }H(\mathbb{X})H(\mathbb{X})
\]
which is tantamount to saying that
\begin{equation} \label{eq:equality_cov}
 1 - \sum_{r\geq 1} \E{\omega_{r,{1}} \omega_{r,2}}
 =   \sum_{r\neq l} \E{\omega_{r,{1}}}\E{\omega_{l,{2}} }.
\end{equation}
Coming back to the evaluation of $\E{G^*_{1}(A)\cdot G^*_{2}(B)}$, we have:
\begin{equation} \label{eq:II_ex_val}
\begin{split}
\E{G^*_{1}(A)\cdot G^*_{2}(B)} &= \sum_{r\geq 1} \E{\omega_{r,{1}}}\E{\omega_{r,{2}}}
H(A\cap B) +   \sum_{r\neq l} \E{\omega_{r,{1}}}\E{\omega_{l,{2}} }H(A)H(B)\\
& = \sum_{r\geq 1} \{\E{\omega_{r,{1}}}\}^2
H(A\cap B) +   \left( 1- \sum_{r\geq 1} \{\E{\omega_{r,{1}}}\}^2 \right)H(A)H(B),
\end{split}
\end{equation}
where we used \eqref{eq:equality_cov} and the fact that $\omega_{r,1}$ and $\omega_{r,2}$ are independent and identically distributed.
It remains to evaluate the infinite series over $r \geq 1$ in \eqref{eq:II_ex_val}, and this issue may be easily addressed, indeed: 
\begin{align*}
\sum_{r\geq 1}\{\E{\omega_{r,{1}}}\}^2 &= \sum_{r\geq 1}\left\{ \E {V_r\prod_{q=1}^{r-1}(1-V_q)} \right\}^2\\
&=\sum_{r\geq 1}\left[ \frac{1}{ (1+\beta)^2 }  \left(\frac{\beta}{1+\beta}\right)^{2({r-1})}\right]=
\frac{1}{2\beta+1}.
\end{align*}
Substituting the previous expression in \eqref{eq:II_ex_val}, we get:
\begin{align*}
\E{G^*_{1}(A) G^*_{2}(B)} = \frac{1}{1+2\beta} 
H(A\cap B) + \frac{2\beta}{1+2\beta} H(A)H(B).
\end{align*}
Putting the expressions of $\E{G_1^*(A) G_2^*(B)}$ and $\E{G_1^*(A)G_1^*(B)}$ in \eqref{eq:EG1G2}, we obtain
\begin{equation} \label{eq:EG1G2_final}
\begin{split}
& \E{G_1(A)  G_2(B)}\\
& \qquad= H(A \cap B)\left( \frac{q_1}{1+\beta}+\frac{1-q_1}{1+2\beta} \right)+
H(A) H (B) \left( q_1 \frac{\beta}{1+\beta}+(1-q_1)\frac{2\beta}{1+2\beta} \right),
\end{split}
\end{equation}
where we recall that $q_1=\frac{1}{\alpha+1}$.
We can use \eqref{eq:EG1G2_final} to evaluate the covariance between $G_j(A)$ and $G_{j'}(A)$ for $j \not = j'$:
\begin{align*}
& Cov\left(G_j(A),G_{j'}(B)\right) \\
& \qquad= H(A \cap B)\left( \frac{q_1}{1+\beta}+\frac{1-q_1}{1+2\beta} \right)+
H(A) H (B) \left( q_1 \frac{\beta}{1+\beta}+(1-q_1)\frac{2\beta}{1+2\beta} -1\right)\\
& \qquad= H(A \cap B)\left( \frac{q_1}{1+\beta}+\frac{1-q_1}{1+2\beta} \right)+
H(A) H (B) \left( -\frac{q_1}{1+\beta}-\frac{1-q_1}{1+2\beta} \right),
\end{align*}
hence \eqref{eq:cov} is now proved.

As for the determination of the correlation \eqref{eq:corr_same}, we first specialize \eqref{eq:cov} when $A=B \in \mathcal{X}$, to get:
\begin{equation}  \label{eq:cov_uguali}
Cov\left(G_j(A),G_{j'}(A)\right) =  \left( \frac{q_1}{1+\beta}+\frac{1-q_1}{1+2\beta} \right) H(A)(1-H(A)),
\end{equation}
and then we divide  $Cov\left(G_j(A),G_{j'}(A)\right)$ by the squared roots of the variances $Var(G_j(A)$ and $Var(G_{j'}(A)$. 
More precisely we have:
\begin{equation} \label{eq:corr1}
\begin{split}
   Corr\left(G_j(A),G_{j'}(A)\right) &=  \frac{Cov\left(G_j(A),G_{j'}(A)\right)}{\sqrt{Var(G_j(A))\cdot Var(G_{j'}(A))}} \\
   & \stackrel{\eqref{eq:cov_uguali}}{=}
  \frac{H(A)(1-H(A))}{\sqrt{Var(G_j(A))\cdot Var(G_{j'}(A))}} \left( \frac{q_1}{1+\beta}+\frac{1-q_1}{1+2\beta} \right) .
  \end{split}
\end{equation}
where the variances in the denominator may be easily evaluated as follows
\begin{align*}
Var(G_j(A))=&\E{G_j(A)^2}-\E{G_j(A)}^2
\\=&
\E{\E{G_j(A)^2|Q}}-\E{G_j(A)}^2
=\E{G^*_1(A)^2}-\E{G_1^*(A)}^2\\
=& \frac{1}{\beta+1} H(A) + \frac{\beta}{1+\beta}H(A)^2-H(A)^2\\
=& \frac{1}{\beta+1} H(A)\left(1-H(A)\right),
\end{align*}
for any $j =1 ,\ldots, J$.
Putting the previous expression in \eqref{eq:corr1} we get:
\begin{align*}
\rho_{j,j'}=Corr(G_j(A),G_{j'}(A))&=\left( 
\frac{q_1}{\beta+1}+\frac{1-q_1}{2\beta+1} \right)\bigg/\frac{1}{\beta+1}\\
&=q_1+\frac{\beta+1}{2\beta+1}(1-q_1)= 1 - \frac{\beta}{2\beta+1}(1-q_1)\\
&= 1 - \frac{\beta}{2\beta+1}\cdot\frac{\alpha}{1+\alpha},
\end{align*}
and \eqref{eq:cov_uguali} is now proved. From the last expression, we finally observe that 
$\rho_{j,j'}$ is always in between $1/2$ and $1$.

\subsection{Proof of Proposition \ref{prp:mixed_moments}}
Recalling the CAM model \eqref{CAM1}-\eqref{CAM2}, we get
\begin{align*}
&\mathbb{E} \left[ \int_{\mathsf{P}_\mathbb{X}^2} f_1(g_1 ) f_2 (g_2) Q(d g_1) Q (d g_2) \right]\\
& \qquad\qquad  = \mathbb{E} \left[ \int_{\mathsf{P}_\mathbb{X}^2} f_1(g_1 ) f_2 (g_2) \sum_{k_1\geq 1} \pi_{k_1} 
\delta_{G_{k_1}^*} (d g_1) \sum_{k_2\geq 1} \pi_{k_2} \delta_{G_{k_2}^*} (d g_2)\right]\\
& \qquad\qquad = \mathbb{E} \left[ \int_{\mathsf{P}_\mathbb{X}^2} f_1(g_1 ) f_2 (g_2) \sum_{k\geq 1} \pi_{k}^2
\delta_{G_{k}^*} (d g_1) \delta_{G_k^*}(d g_2)\right]   \\
& \qquad\qquad\qquad + \mathbb{E} \left[ \int_{\mathsf{P}_\mathbb{X}^2} f_1(g_1) f_2 (g_2) \sum_{k_1 \not = k_2} 
\pi_{k_1} \pi_{k_2}  \delta_{G_{k_1}^*}(d g_1) \delta_{G_{k_2}^*}(d g_2)  \right].
\end{align*}
Observe that the $G_k^*$'s are all Dirichlet processes having the same law on the space $\mathsf{P}_\mathbb{X}$, which will be denoted by $\mathcal{P}$, depending on the total mass $\alpha$ and the base measure $H$. We also point out that the $G_k^*$'s are not independent random elements for different values of $k$, indeed they share the same random atoms $(\theta_l)_{l \geq 1}$, nevertheless if $k_1 \not = k_2 $, the distribution of  $(G_{k_1}^*, G_{k_2}^*)$ equals the distribution of $(G_1^*,G_2^*)$, which will be denoted by $\mathcal{P}_{[2]}$. Therefore, by applying the Tonelli--Fubini Theorem, we obtain 
\begin{align*}
&\mathbb{E} \left[ \int_{\mathsf{P}_\mathbb{X}^2} f_1(g_1 ) f_2 (g_2) Q(d g_1) Q (d g_2) \right]\\
& \qquad\qquad = \sum_{k \geq 1}  \mathbb{E} \pi_{k}^2   \mathbb{E} \int_{\mathsf{P}_\mathbb{X}^2} f_1 (g_1) f_2 (g_2)  \delta_{G_k^*} (d g_1) \delta_{G_k^*} (d g_2)\\
& \qquad\qquad\qquad +  \sum_{k_1 \not = k_2} \mathbb{E}
\pi_{k_1} \pi_{k_2}  \mathbb{E} \int_{\mathsf{P}_\mathbb{X}^2} f_1(g_1) f_2 (g_2) \delta_{G_{k_1}^*}(d g_1) \delta_{G_{k_2}^*}(d g_2)  \\
& \qquad\qquad = q_1  \int_{\mathsf{P}_\mathbb{X}^2} f_1 (g) f_2 (g) \mathcal{P} (d g) + (1-q_1)
\int_{\mathsf{P}_\mathbb{X}^2} f_1 (g_1) f_2 (g_2) \mathcal{P}_{[2] } (d g_1, d g_2 ),
\end{align*}
and then the thesis follows.\\

\subsection{Proof of Theorem \ref{thm:EPPF}}
We first evaluate the expected value in the definition of pEPPF \eqref{eq:pEPPF_def}, for $J=2$,
\begin{align*}
\mathbb{E}  \prod_{j=1}^2 \prod_{i=1}^s G_j^{n_{i,j}} (d y_i^*) & = 
\mathbb{E} \left[ \mathbb{E} \left[ \prod_{j=1}^2 \prod_{i=1}^s G_j^{n_{i,j}} (d y_i^*) \Big| Q  \right] \right]\\
& = \mathbb{E} \left[ \int_{\mathsf{P}_\mathbb{X}^2} \prod_{j=1}^2  \prod_{i=1}^s  g_j^{n_{i,j}} (d y_i^*) Q(d g_1) Q(d g_2)
\right].
\end{align*}
Now we apply  Equation \eqref{eq:mixed} to the previous integral
where the functions $f_j$, as $j=1,2$, are defined by
\[
f_j (g_j) := \prod_{i=1}^s g_j^{n_{i,j}} (d y_i^*),
\]
and then we get
\begin{equation}
\mathbb{E}  \prod_{j=1}^2 \prod_{i=1}^s G_j^{n_{i,j}} (d y_i^*) =
q_1 \mathbb{E} \prod_{j=1}^2 \prod_{i=1}^s (G_1^*)^{n_{i,j}} (d y_i^*)  +
(1-q_1) \mathbb{E} \prod_{j=1}^2 \prod_{i=1}^s (G_j^*)^{n_{i,j}} (d y_i^*).
\end{equation}
We finally integrate over the space $\mathbb{X}^s$ to get the result, i.e. \eqref{eq:pEPPF_CAM}.\\

\subsection{Proof of Proposition \ref{prp:noEX}}
Assume that the two samples $\bm{y}_1$ and $\bm{y}_2$ share $s_0>0$ distinct values denoted here as 
$y_{1,0}^*, \ldots y_{s_0, 0}^*$ with frequencies $(q_{1,j}, \ldots , q_{s_0, j})$ in the $j$-th sample, as $j=1,2$. We further suppose that the $j$-th sample contains exactly $s_j$ distinct 
observations not shared with the other one, and denoted here by $y_{1,j}^*, \ldots , y_{s_j, j}^*$, 
 for $j=1,2$; besides the vector of corresponding frequencies will be denoted as $(r_{1,j}, \ldots , r_{s_j,j})$. We obviously have that $s=s_0+s_1+s_2$.\\
Using the representation of the $G_k^*$'s in the CAM model \eqref{CAM1}--\eqref{CAM2}, we get
\begin{equation*}
 \mathbb{E} \prod_{j=1}^2 \prod_{i=1}^s (G_j^*)^{n_{i,j}} (d y_i^*)
= \mathbb{E} \prod_{j=1}^2 \prod_{i=1}^s \left(\sum_{l\geq 1} \omega_{l,j}\delta_{\theta_l} (d y_i^*)\right)^{n_{i,j}}.
\end{equation*}
Exploiting the partition of the data described at the beginning of the proof, we obtain
\begin{align*}
& \mathbb{E} \prod_{j=1}^2 \prod_{i=1}^s (G_j^*)^{n_{i,j}} (d y_i^*) \\
& \qquad = \mathbb{E} \prod_{j=1}^2 \prod_{i=1}^{s_j} \left(\sum_{l\geq 1} \omega_{l,j}^{r_{i,j}}\delta_{\theta_l} (d y_{i,j}^*)\right)
\prod_{i=1}^{s_0} \left( \sum_{l\geq 1}\omega_{l,1}^{q_{i,1}}\omega_{l,2}^{q_{i,2}} \delta_{\theta_l} (d y_{i,0}^*)\right)
+ o \left(  \prod_{j=0}^2 \prod_{i=1}^{s_j} H (d y_{i,j}^*) \right)\\
& \qquad = \sum_{\not =} 
\mathbb{E} \left[\prod_{j=1}^2 \prod_{i=1}^{s_j}  \omega_{l_{i,j},j}^{r_{i,j}}
\prod_{i=1}^{s_0} \omega_{l_{i,0},1}^{q_{i,1}}\omega_{l_{i,0},2}^{q_{i,2}} \right] \prod_{j=0}^2 \prod_{i=1}^{s_j} H (d y_{i,j}^*) 
+ o \left(  \prod_{j=0}^2 \prod_{i=1}^{s_j} H (d y_{i,j}^*) \right).
\end{align*} 
where the sum $\sum_{\not =} $ is extended over all possible values  of the distinct natural numbers $\{  l_{i,j} : \; i=1, \ldots , s_j ,\; j=0,1,2 \}$. Integrating over  $\mathbb{X}^s$ we get that
\begin{equation} \label{eq:2term_pEPPF}
\int_{\mathbb{X}^s} \mathbb{E} \prod_{j=1}^2 \prod_{i=1}^s (G_j^*)^{n_{i,j}} (d y_i^*)=
\sum_{\not =} 
\mathbb{E} \left[\prod_{j=1}^2 \prod_{i=1}^{s_j}  \omega_{l_{i,j},j}^{r_{i,j}}
\prod_{i=1}^{s_0} \omega_{l_{i,0},1}^{q_{i,1}}\omega_{l_{i,0},2}^{q_{i,2}} \right]
\end{equation}
which is positive whenever $s_0>0$.
 
\section{Truncated Blocked Gibbs Sampler for CAM}  \label{app:algo}

The posterior distribution is analytically intractable, which forces us to develop sampling algorithms to simulate from it. A P\'olya Urn representation would be too expensive in computational cost. Instead, we provide two different algorithms: a Blocked Gibbs sampler \citep{Ishwaran2001a}, mimicking the one proposed in \citep{Rodriguez2008} and a nested slice sampler \citep{Damien1999,Walker2007,Kalli2011}. Here we discuss the former one. The Truncated CAM model has the following form:
\begin{equation}
\begin{aligned}
y_{i,j}|\bm{ M, \theta} & \sim  N\left(\cdot | \theta_{M_{i,j}}\right), \quad
&M_{i,j}|\bm{S,\omega}  \sim \sum_{l=1}^{L} \omega_{l, S_j} \delta_l(\cdot),\\
\bm{\omega}_k | \bm{S}=\bm{\omega}_k &\sim GEM(\alpha), \quad
&S_{j}|\bm{\pi}  \sim \sum_{k=1}^{K} \pi_{k} \delta_k(\cdot),\\
\bm{\pi} &\sim GEM(\beta), \quad
&\theta_{l} \sim \pi(\theta_{l}).
\end{aligned}
\label{MembershipTRUNC}
\end{equation}
The Truncated version of CAM (TCAM) \eqref{MembershipTRUNC} can be extended to a Truncated version of DCAM (TDCAM) once the likelihood is modified according to \eqref{newLik}, In the following we report the Gibbs Sampler for the TDCAM, the extension of the sampler to accomodate the presence of a covariate linearly introduced.
Notice that some of the conditioning variables are collapsed \citep{Liu1994}, to enhance the speed of convergence and the mixing of the chains.

\subsection{TDCAM: Gibbs Sampler}
\label{tdcam}
Denote with $\bm{V}$ the vector containing all the variables of model \eqref{MembershipTRUNC}, and let $\bm{V}^{-\bm{s}}$  be the same vector $\bm{V}$ with the variable $\bm{s}$ removed.\\
The steps of the MCMC are the following:

\begin{enumerate}
	\item The full conditional for each $y_{i,j}$ is Truncated Normal, with support $\left[a_{z_{i,j}},a_{z_{i,j}+1} \right)$:
	$$
p\left(y_{i,j}|\bm{V}
	\right) \sim TN(\mu_{M_{i,j}} ,\sigma^2_{M_{i,j}};a_{z_{i,j}},a_{z_{i,j}+1} ).
	$$
	This can be easily done with the help of the R package \texttt{TruncatedNormal}, which relies on a recently improved algorithm exploiting minmax tilting \citep{Botev2017}.
	
	\item The full conditional for the observational cluster labels $M_{i,j}$, once the latent variable $\bm{y}$ is integrated out, is a discrete distribution, given by
	\begin{align*}
	p\left(M_{i,j}=l|\bm{V}^{-\bm{y}}\right) &\propto \omega_{l, S_j} \Delta\Phi\left(a_{z_{i,j}};\mu_{M_{i,j}},\sigma^2_{M_{i,j}} \right),
	\end{align*}
	for any $i=1, \ldots , n_j$, $j=1, \ldots , J$.
	\item The full conditional for the distributional cluster labels $S_{j}$ is given by:
	\begin{equation*}
	p\left(S_{j}=k|\bm{V}^{-(\bm{y},\bm{M})}\right) \propto
	\pi_{k} \prod_{i=1}^{n_j}\Bigg(  \sum_{m=1}^{L}  \omega_{m, k}\Delta\Phi\left(a_{y_{i,j}};\mu_{m},\sigma^2_{m} \right)\Bigg),
	\end{equation*}
	for any $j=1, \ldots , J$.
	\item To sample the full conditional of the weights $\bm{\pi}$ at the distributional level, we first need to define $m^*_k$ as the number of groups assigned to the same distributional cluster $k$, where $\sum_{k=1}^{K}m^*_k=J$ the total number of observed groups. Then,
	\begin{align*}
	p\left(\bm{\pi}|\bm{V}\right)  &\propto p\left(\bm{S}|\bm{\pi}\right)p\left(\bm{\pi}\right) \propto p\left(\bm{\pi}\right) \pi_1^{m^*_1} \cdots \pi_K^{m^*_K}.
	\end{align*}
	Referring to the Stick Breaking representation, the  full conditional of the different sticks $v_k$, as $ k=1,\ldots,K $, equals:
	\[v_k \sim Beta\left(1+m_k^*, \beta+ \sum_{s=k+1}^{K}m_k^*\right).\]
	
	\item The derivation of the full conditional for $\bm{\omega}$ is similar, even it requires more care. We have
	\begin{align*}
	p\left(\bm{\omega}|\bm{V}\right)  \propto p\left(\bm{M}|\bm{S},\bm{\omega},\bm{\xi_0}\right)p\left(\bm{\omega}\right) \propto \prod_{k=1}^{K}p\left(\bm{\omega_k}\right) \prod_{j=1}^J\prod_{i=1}^{n_j}  \left( \sum_{l=1}^{L} \omega_{l, S_j} \delta_l(\cdot) \right).
	\end{align*}
	The previous formula can be decomposed into the product of $K$ elements and we can focus only on the case $S_j=k$. Let us define $n_{l,k}$ as the total number of observations assigned to the distributional cluster $k$ in the observational group $l$. The full conditional has this Stick-Breaking representation for $u_{l,k}, \forall k$:
	\[u_{l,k} \sim Beta\left(1+n_{l,k}, \alpha+ \sum_{r=l+1}^{L} n_{r,k}\right), \:\:\: l=1,\ldots,L.\]
	
	\item Let us define $n_{l, \cdot}=\sum_{k=1}^{K}n_{l,k}$ and  
	\[
	\bar{y}_{l,\cdot} := \frac{1}{n_{l,\cdot}}\sum_{i,j: M_{i,j}=l}y_{i,j}.
	\]
	Exploiting the conjugacy property, we obtain the full conditional for $\theta_l=\left(\mu_l,\sigma^2_l\right)$:
	\begin{align*}
	\left(\mu_{l},\sigma^2_{l}\right)|\bm{V} &\sim NIG\left(m_0^*,\kappa_0^*,\alpha_0^*,\beta_0^* \right).
	\end{align*}
	where
	\begin{equation*}
	m_0^*=\frac{\kappa_0 m_0+n_{l,\cdot}\bar{y}_{l,\cdot}}{\kappa_0+n_{l,\cdot}} \quad
	\kappa_0^*=\kappa_0 + n_{l,\cdot} \quad \alpha_0^* = \alpha_0+n_{l,\cdot}/2
	\end{equation*}
	and
	\begin{align*}
	\beta_0^* = \beta+0.5\left(\sum_{ij:\, M_{i,j}=l} \left(y_{i,j}-\bar{y}_{l,\cdot}\right)^2+\left(\frac{\kappa_0 n_{l,\cdot}}{\kappa_0+n_{l,\cdot}}\right)\left(y_{l,k}-m_0\right)^2\right).\\
	\end{align*}
	\item In case the precision parameters $\alpha$ and $\beta$ of the two DPs are assumed stochastic, distributed as $Gamma\left(a_\alpha,b_\alpha\right)$ and $Gamma\left(a_\beta,b_\beta\right)$, we can still exploiting conjugacy. The full conditionals distributions are:
	\begin{align*}
	\alpha|\bm{V}  \sim & Gamma\left( a_\alpha + (K-1),     b_\alpha -  \sum_{k=1}^{K-1}\log(1-v_k)\right) ,\\
	\beta|\bm{V} \sim & Gamma\left( a_\beta  + K\cdot(L-1), b_\beta  -  \sum_{l=1}^{L-1}\sum_{k=1}^{K} \log(1-u_{l,k})\right)	.
	\end{align*}	
\end{enumerate}
Notice that we naturally set $a_{y_{i,j}}=y_{i,j}$. As suggested in \citep{Rodriguez2008}, each step of this algorithm can be parallelized, in order to gain computational speed.

\subsection{Linearly incorporating a covariate in the Likelihood}
\label{covariate}

If we want to linearly add regressor to the mean, we update model \eqref{MembershipTRUNC} simply assuming:
\begin{equation}
z_{i,j}|y_{i,j} \sim \sum_{g=0}^{+\infty}\delta_g(\cdot) \mathbf{1}_{ \left[a_g,a_{g+1}\right)}\left(y_{i,j}\right) \quad \quad
y_{i,j}|\bm{ M, \mu,\sigma^2}  \sim  N\left(\mu_{M_{i,j}}+\beta X_j,\sigma^2_{M_{i,j}}\right).
\label{newLikLS111}
\end{equation}
We espouse such representation because of its interpretability: the latent continuous random variable $y_{i,j}$ can be decomposed as $y_{i,j}=\mu_{M_{i,j}} + \beta X_j+ \varepsilon_{i,j},$ where $\varepsilon\sim N\left(0,\sigma^2_{M_{i,j}}\right)$. In other words, we model the every single latent value as the sum of an effect specific for each observational cluster, an effect due to the regressor value of each individual multiplied by a overall coefficient and a completely random effect, whose entity still depends on the observational cluster. This choice does not complicate the algorithm presented in the previous section: the full conditionals 1-3 are preserved if the mean is modified accordingly, switching from $\mu_{M_{i,j}}$ to $\mu_{M_{i,j}}+\beta X_j$. Step 6 remains the same once we substitute $y_{i,j}$ with $d_{i,j}=y_{i,j}-\beta X_j$. Steps 4, 5 and 7 are not affected by this change.\\
Finally, if we assume $\beta \sim N\left(m_\beta, \frac{1}{\kappa_\beta}\right)$, we can perform inference on the introduced coefficient.
Define $R^1=\sum_{i,j}\frac{X_j^2}{\sigma^2_{M_{i,j}}}$ and $R^2=\sum_{i,j}\frac{d_{i,j}\cdot X_j}{\sigma^2_{M_{i,j}}}$. The full conditional for $\beta$ is:
\begin{equation*}
\beta|\bm{V}\sim N\left( \frac{m_\beta \kappa_\beta + R^2}{\kappa_\beta+R^1}, \frac{1}{\kappa_\beta+R^1}    \right).
\end{equation*}
This framework can be easily extended to accommodate for the presence of multiple covariates.

\section{Error bounds in total variation distance} 

In Section \ref{app:algo} we have depicted a truncated blocked Gibbs sampler, we now evaluate the truncation error arising from these algorithms 
for the CAM model (Section \ref{app:error_CAM}) and the CAMM (Section \ref{app:error_CAMM}). The errors between the random distribution and its truncated version will be evaluated using the total variation distance.
For the reader's convenience  we recall that if  $P, Q \in \mathsf{P}_{\mathbb{X}}$ are  probability measures defined on  $\left(\mathbb{X},\mathcal{X}\right)$, the  distance in total variation between $P$ and $Q$ is defined as
$$d_{TV}(P,Q) = \sup_{A\in \mathcal{X}}|P(A)-Q(A)|.$$
If $P,Q$ are absolutely continuous w.r.t. a measure $\mu$ then it can be expressed as
$$d_{TV}(P,Q) =\frac12 \int_{\mathcal{X}}  \biggl|\frac{dP}{d\mu}-\frac{dQ}{d\mu}\biggr|d\mu.$$
Moreover, if $\mathcal{X}$ is a discrete space or if $P$ and $Q$ are concentrated on a countable set $\Omega\subset\mathcal{X}$ then
$$d_{TV}(P,Q) =\frac12 \sum_{x\in \Omega} \biggl|P(x)-Q(x)\biggr| .$$


\subsection{Truncation error in CAM} \label{app:error_CAM}

In this section we quantify the error committed when we replace  the random probability measures $G_j$ with the corresponding truncated versions.
We recall that $G_1, \ldots , G_J | Q \stackrel{i.i.d.}{\sim} Q$, where 
$Q$ has been defined in \eqref{CAM2}:
$$ Q=\sum_{k\geq 1}\,  \pi_k \, \delta_{G^*_k }, \quad
G_k^* = \sum_{l \geq 1} \omega_{l,k} \delta_{\theta_l}. $$
In order to formally define the truncated versions of $G_1, \ldots , G_J$, we exploit  the latent random variable $\xi_j|Q \stackrel{iid}{\sim} \sum_{k=1}^{+\infty} \pi_k \delta_k $, as $j=1 , \ldots , J$, which  identifies the mixture component from which $G_j$ is generated, conditionally on $Q$. Thus, conditionally on the value $\xi_j = k$, the truncated random probability measures associated to each $G_j$ are formally defined as follows
\begin{equation} \label{eq:truncation_Gj}
G^{(K,L)}_j =\begin{cases}
& \sum_{l=1}^{L} \omega_{l,k}^{(K,L)}\delta_{\theta_l} \quad \quad \text{if  } \xi_j \leq K\\
&\sum_{l=1}^{L} \omega_{l,K}^{(K,L)}\delta_{\theta_l} \quad \quad \text{if  } \xi_j > K
\end{cases}  
\end{equation}
and 
\begin{align*}
&\omega_{l,k}^{(K,L)} = 	\omega_{l,k} \quad \quad \quad \quad \text{if  } l \leq L-1, \quad \text{and} \quad 
\omega_{L,k}^{(K,L)} =  1-\omega_{1,k}-\ldots-\omega_{L-1,k} \\
&\pi_{k}^{(K,L)} = \pi_k \quad \quad \quad \quad \text{if  } k \leq K-1, \quad  \text{and} \quad  \pi_{K}^{(K,L)} = 1-\pi_1-\ldots-\pi_{K-1}
\end{align*}
where $K, L>0$ define the truncation levels for the different random probability measures.
\begin{proposition} \label{prp:DTV_CAM}
Let $G_j | Q \sim Q$ and $G_j^{(K,L)}$ the truncation of $G_j$ defined in 
\eqref{eq:truncation_Gj}, then the expected value of the distance in total variation between them can be estimated as follows:
\begin{equation} \label{eq:dTV}
\E{d_{TV}\left(G_j,G_j^{(K,L)} \right) } \leq 	\left(1-\left(\frac{\alpha}{1+\alpha}\right)^K\right)\left(\frac{\beta}{1+\beta}\right)^L+\left(\frac{\alpha}{1+\alpha}\right)^K,
\end{equation}
for any $j=1,\ldots , J$.
\end{proposition}
\begin{proof}
First of all observe that, conditioning on $\xi_j =k$, we recognize two distinct situations to upper bound the total variation distance between 
$G_j$ and its truncated counterpart as described below.
\begin{enumerate}
    \item If $\xi_j =k \leq K$, then we have
\begin{align*}
d_{TV}(G_j,G_j^{(K,L)}) &= \frac12 \left( \sum_{l=1}^{L} |\ot{l}{k}-\ot{l}{k}^{(K,L)}| + |\ot{l}{k}-\ot{l}{k}^{(K,L)}|+\sum_{l\geq L+1} |\ot{l}{k}-0|  \right) \\ &=
\frac12 \left( |\ot{L}{k}-1 + \ot{1}{k}+\ldots+\ot{L-1}{k}|+\sum_{l\geq L+1} \ot{l}{k}  \right)\\&=
\frac12 \left( 1-\sum_{l=1}^{L}\ot{l}{k}  +\sum_{l\geq L+1}\ot{l}{k}  \right) = \left( 1-\sum_{l=1}^{L}\ot{l}{k}\right).
\end{align*}
\item If $\xi_j >K$, we use the following trivial upper bound
$d_{TV}(G_j,G_j^{(K,L)}) \leq 1$.
\end{enumerate}
In light of the previous considerations, we are now ready to compute
\begin{align*}
\mathbb{E}\left[d_{TV}\left(G_j,G_j^{(K,L)} \right) \right] &= \E{\E{d_{TV}\left(G_j,G_j^{(K,L)} \right)|\xi_j,Q }}\\
& =
\E{ 	\sum_{k=1}^{K} \pi_k \E{d_{TV}\left(G_j,G_j^{(K,L)} \right)|\xi_j=k,Q }}\\
& \qquad\qquad\qquad+
	\E{ \sum_{k=K+1}^{+\infty} \pi_k  \E{  d_{TV}\left(G_j,G_j^{(K,L)} \right)|\xi_j=k,Q }}
\\ 
&  \leq  \E{ 
	\sum_{k=1}^{K} \pi_k \left(1-\sum_{l=1}^{L}\ot{l}{k}\right) +  \sum_{k=K+1}^{+\infty} \pi_k } \\ 
	&  =
\E{ 
	\sum_{k=1}^{K} \pi_k } \cdot \E{ \left(1-\sum_{l=1}^{L}\ot{l}{k}\right) } + \E{ \sum_{k=K+1}^{+\infty} \pi_k }  \\ 
	& \leq
\E{ \left(1-\sum_{l=1}^{L}\ot{l}{k}\right) } + \E{ 1-\sum_{k=1}^{K} \pi_l }  \\ 
&  =
\left(1-\left(\frac{\alpha}{1+\alpha}\right)^K\right)\left(\frac{\beta}{1+\beta}\right)^L+\left(\frac{\alpha}{1+\alpha}\right)^K
\end{align*}
where the last equality follows by straightforward calculations based on the stick--breaking representation of the weights.

\end{proof}

\subsection{Approximation Error in Mixture Models (CAMM)}
\label{app:error_CAMM}

Consider $J$ groups, each of them containing $n_j$ observations, $j=1,\ldots,J$. Denote by $\bm{y}_j=\left(y_{1,j},\ldots,y_{n_j,j}\right)$ for $j=1,\ldots,J$ the observations from the $j$-th component of the mixture model $y_{i,j}|\theta_{i,j}\sim f(\cdot|\theta_{i,j})$ with $\theta_{i,j}|G_1,\ldots,G_J \sim G_j$ where the $G_j$'s are generated according to a CAM. We suppose that $\theta_{i,j} \in \Theta$, where $\Theta$ is a Polish space equipped with its corresponding Borel $\sigma$--field $\mathcal{T}$.
We further denote by $\bm{y} = (\bm{y}_1, \ldots , \bm{y}_J)$ the vector containing all the observations.
We would like to upper bound the distance in total variation between the law of the data $\bm{y}$
\[ \pi\left(\bm{y}\right) = \E{ \prod_{j=1}^{J}\prod_{i=1}^{n_j} \int_{\Theta} f\left(y_{i,j}|\theta_{i,j}\right)G_j(d\theta_{i,j})   },   \]
and the law of the data $\pi^{(K,L)}$ when the random probability measures $G_j$'s are replaced with the corresponding truncated version $G_j^{(K,L)}$
defined in \eqref{eq:truncation_Gj}, i.e.
\[ \pi^{(K,L)}\left(\bm{y}\right) = \E{ \prod_{j=1}^{J}\prod_{i=1}^{n_j} \int_{\Theta} f\left(y_{i,j}|\theta_{i,j}\right)G^{(K,L)}_j(d\theta_{i,j})   }.   \]

\begin{proposition}
The distance in total variation between $\pi$ and  $\pi^{(K,L)}$ satisfies
\begin{equation} \label{eq:distTV_mixture}
    d_{TV}\left(\pi,\pi^{(K,L)}\right) \leq 
    N\left[ \left(\frac{\beta}{1+\beta}\right)^L+\left(\frac{\alpha}{1+\alpha}\right)^K\right],
\end{equation}
where $N= n_1 +\cdots + n_J$.
\end{proposition}
\begin{proof}
The distance  $d_{TV}\left(\pi,\pi^{(K,L)}\right)$ can be evaluated as follows:
\begin{align*}
& d_{TV}\left(\pi,\pi^{(K,L)}\right) \\
&\quad = \frac12 \int_{\mathbb{X}^N} \bigg|\frac{d\pi}{d\bm{y}}-\frac{d\pi^{K,L}}{d\bm{y}}\bigg|d\bm{y} \\ 
&\quad = \frac12 \int_{\mathbb{X}^N} \bigg|\E{ 
	\prod_{j=1}^{J}\prod_{i=1}^{n_j} \int_{\Theta} f\left(y_{i,j}|\theta_{i,j}\right)G_j(d\theta_{i,j})-
	\prod_{j=1}^{J}\prod_{i=1}^{n_j} \int_{\Theta} f\left(y_{i,j}|\theta_{i,j}\right)G^{(K,L)}_j(d\theta_{i,j})  
}\bigg| d\bm{y} \\ 
& \quad =
\frac12 \int_{\mathbb{X}^N} \bigg|\mathbb{E} \bigg[
	\int_{\Theta^N}
	\prod_{j=1}^{J}\prod_{i=1}^{n_j} f\left(y_{i,j}|\theta_{i,j}\right)
	\prod_{j=1}^{J}\prod_{i=1}^{n_j} G_j(d\theta_{i,j})\\
	& \quad \qquad\qquad\qquad-
	\int_{\Theta^N}\prod_{j=1}^{J}\prod_{i=1}^{n_j} f\left(y_{i,j}|\theta_{i,j}\right)
	\prod_{j=1}^{J}\prod_{i=1}^{n_j}G^{(K,L)}_j(d\theta_{i,j})  
\bigg]\bigg| d\bm{y} .
\end{align*}
By an application of the  Tonelli-Fubini theorem and the Jensen inequality, we obtain
\begin{align*}
& d_{TV}\left(\pi,\pi^{(K,L)}\right) \\
& \quad =
\frac12 \int_{\mathbb{X}^N} \bigg|\int_{\Theta^N} \prod_{j=1}^{J}\prod_{i=1}^{n_j} f\left(y_{i,j}|\theta_{i,j}\right)
\E{
	\prod_{j=1}^{J}\prod_{i=1}^{n_j} G_j(d\theta_{i,j})-
	\prod_{j=1}^{J}\prod_{i=1}^{n_j}G^{(K,L)}_j(d\theta_{i,j})  
}\bigg| d\bm{y} \\ 
& \quad  \leq 
\frac12 \int_{\mathbb{X}^N} \int_{\Theta^N} \prod_{j=1}^{J}\prod_{i=1}^{n_j} f\left(y_{i,j}|\theta_{i,j}\right)d\bm{y}
\bigg| \E{
	\prod_{j=1}^{J}\prod_{i=1}^{n_j} G_j(d\theta_{i,j})-
	\prod_{j=1}^{J}\prod_{i=1}^{n_j}G^{(K,L)}_j(d\theta_{i,j})  
}\bigg|  \\ 
& \quad  =
\frac12 \int_{\Theta^N} \underbrace{\int_{\mathbb{X}^N}  \prod_{j=1}^{J}\prod_{i=1}^{n_j} f\left(y_{i,j}|\theta_{i,j}\right)d\bm{y}}_{=1}
\bigg| 
	\underbrace{\E{\prod_{j=1}^{J}\prod_{i=1}^{n_j} G_j(d\theta_{i,j})}}_{=:m}-
	\underbrace{\E{\prod_{j=1}^{J}\prod_{i=1}^{n_j}G^{(K,L)}_j(d\theta_{i,j})}}_{=:m^{(K,L)}}\bigg|  \\ 
& \quad =
\frac12 \int_{\Theta^N} 
\bigg| \E{
	\prod_{j=1}^{J}\prod_{i=1}^{n_j} G_j(d\theta_{i,j})-
	\prod_{j=1}^{J}\prod_{i=1}^{n_j}G^{(K,L)}_j(d\theta_{i,j}) 
}\bigg|  = d_{TV}\left(m,m^{(K,L)}\right) \\
& \quad =
\sup_{A_{i,j} \in \mathcal{T}} \bigg| \E{
	\prod_{j=1}^{J}\prod_{i=1}^{n_j} G_j(A_{i,j})-
	\prod_{j=1}^{J}\prod_{i=1}^{n_j}G^{(K,L)}_j(A_{i,j}) 
}\bigg|.
\end{align*}
We can exchange the expected value with the supremum to get:
\begin{align*}
d_{TV}\left(\pi,\pi^{(K,L)}\right) &\leq 
\sup_{A_{i,j} \in \mathcal{T}}  \E{\bigg|
	\prod_{j=1}^{J}\prod_{i=1}^{n_j} G_j(A_{i,j})-
	\prod_{j=1}^{J}\prod_{i=1}^{n_j}G^{(K,L)}_j(A_{i,j}) 
	\bigg|}\\
 & \leq 
\E{\sup_{A_{i,j} \in \mathcal{T}}\bigg|
	\prod_{j=1}^{J}\prod_{i=1}^{n_j} G_j(A_{i,j})-
	\prod_{j=1}^{J}\prod_{i=1}^{n_j}G^{(K,L)}_j(A_{i,j}) 
	\bigg|}.
\end{align*}
We now apply \citep[Lemma 1, pg. 358]{Billy} to obtain
\begin{align*}
 d_{TV}\left(\pi,\pi^{(K,L)}\right)  & \leq 
\E{\sup_{A_{i,j} \in \mathcal{T}}
	\sum_{j=1}^{J}\sum_{i=1}^{n_j} \bigg|G_j(A_{i,j})-G^{(K,L)}_j(A_{i,j}) 
	\bigg|} \\
	& \leq 
\E{	\sum_{j=1}^{J}\sum_{i=1}^{n_j} \sup_{A_{i,j} \in \mathcal{T}}\bigg|G_j(A_{i,j})-G^{(K,L)}_j(A_{i,j}) 
	\bigg|}
\end{align*}
where we recognize that $\sup_{A_{i,j} \in \mathcal{T}}\bigg|G_j(A_{i,j})-G^{(K,L)}_j(A_{i,j}) 
\bigg|=d_{TV}\left( G_j,G^{(K,L)}_j  \right)$. As a consequence, by an application of Proposition \ref{prp:DTV_CAM}, we get
\begin{align*}
d_{TV}\left(\pi,\pi^{(K,L)}\right) & \leq 
\E{	\sum_{j=1}^{J}\sum_{i=1}^{n_j} d_{TV}\left( G_j,G^{(K,L)}_j  \right)}=
\sum_{j=1}^{J}\sum_{i=1}^{n_j} \E{d_{TV}\left( G_j,G^{(K,L)}_j  \right)}\\
 & \leq  N\left[ \left(\frac{\beta}{1+\beta}\right)^L+\left(\frac{\alpha}{1+\alpha}\right)^K\right]
\end{align*}
and the result follows.
\end{proof}

\section{Additional Details about the Nested Slice Sampler}\label{app:moredetailslice}

As we mentioned, at each iteration we sample among $K^*$ possible distributional cluster labels and $L^{**} = \max\{ L_1^*,\ldots,L^*_{K^*}\}$ possible observational labels.
If $\xi^D_k=\pi_k$ and $\xi^O_{l,k}=\omega_{l,k}$, the values are the lowest integers that ensure, respectively, that
\begin{equation}
\sum_{k=1}^{K^*} \pi_k \geq 1-\min_{j\in \{1,\ldots,J\}} u^D_j\quad \text{and} \quad\sum_{l=1}^{L_k^*} \omega_{l,k} \geq 1-\min_{i \in \{1,\ldots,n_j\}} u^O_{i,j}\quad \forall k=1,\ldots,K^*.
\end{equation}
Instead of relying on the efficient-dependent version, according to \citet{Kalli2011,Hong2016}, we adopt the following geometric deterministic sequences: $\xi^D_k=\left(1-\kappa_D\right)\kappa_D^{k-1}$, and $\xi^O_{l,k}= \xi^O_l=\left(1-\kappa_O\right)\kappa_O^{l-1}$. In this case, it is sufficient to focus only on one observational deterministic sequence, being $\bm{\xi}^O_k$ the same for every $k$. Thus, putting $u^D_{min}=\min_{j}u^D_{j}$ and $u^O_{min}=\min_{i,j}u^O_{i,j}$, we can compute the two thresholds at each MCMC sweep:
\[ K^*= \floor*{\frac{\log\left(u^D_{min}\right)- \log\left(1-\kappa_D\right)}{\log\left(\kappa_D\right)}}, \quad \quad
L^*= \floor*{\frac{\log\left(u^O_{min}\right)- \log\left(1-\kappa_O\right)}{\log\left(\kappa_O\right)}}.
\]
If the precision parameters $\alpha$ and $\beta$ of the two DPs are assumed stochastic and conjugate Gamma distributions are adopted, the full conditionals can be sampled following the procedure proposed in  \citet{Walker2007,Escobar2012}: denote with $c^*$ the number of unique values sampled and with $n$ the number of observations ($n=J$ when the Outer DP is considered, otherwise $n = \sum_{j=1}^{J}n_j$ ). Then the precision parameter of the DP $\gamma$ ($\gamma=\alpha$ when Outer DP, $\gamma=\beta$ otherwise), for both the DPs, can be sampled in two stage, introducing another latent variable $\eta$:
$(a)$ sample $\eta|\gamma,c^* \sim Beta\left(\gamma+1,n\right)$ and $(b)$ sample a new $\gamma$ from the mixture $$\gamma \sim \pi _ { \eta } G ( a + k , b - \log ( \eta ) ) + \left( 1 - \pi _ { \eta } \right) G \left( a + k - 1 , b - \log ( \eta ) \right)$$
where $\pi _ { \eta }= \pi _ { \eta } / \left( 1 - \pi _ { \eta } \right) = ( a + k - 1 ) / \{ n ( b-\log(\eta)\}$.

The exploration of the space of cluster membership labels is a delicate task.
Differently from the marginal specification, where simulation methods are devised in a way that the resulting Markov Chain explores the space of the partitions as equivalence classes over cluster values, a conditional/stick-breaking specification operates on the space of the explicit cluster labels \citep{Porteous2005}.
In this second scenario, it could happen that the chain exploring the cluster membership shows poor mixing, being stuck in one of the local maxima of the posterior. 
To overcome this issue, the label switching moves described in \citep{Papaspiliopoulos2008, Hastie2015} can be added to our setup to improve the mixing.

\section{Additional Plots}
\subsection{Densities of the three scenarios considered in the simulation study}
\begin{figure}[ht!]
    \centering
    \includegraphics[scale=.60]{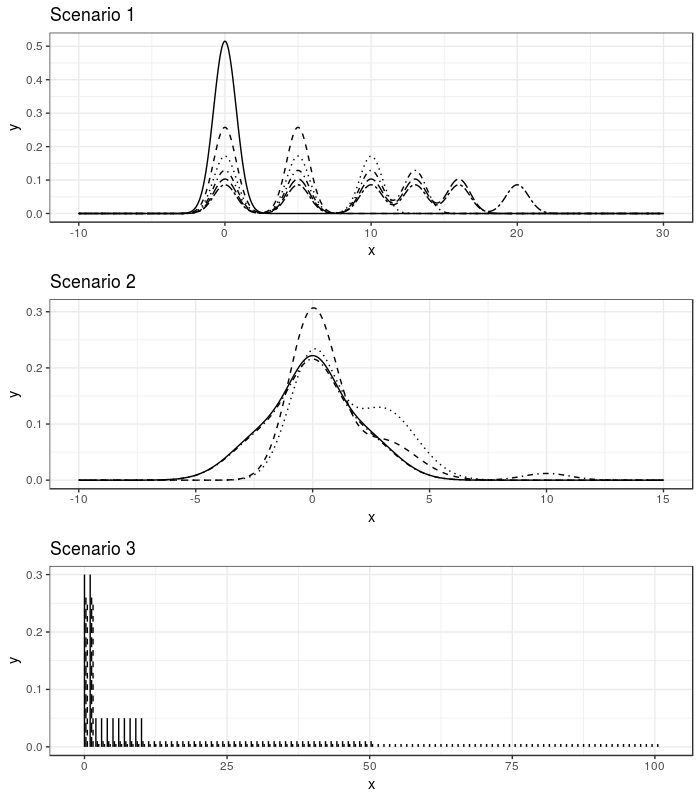}
    \caption{The densities distributions of each unit in three scenarios considered.}
    \label{3sce}
\end{figure}
\FloatBarrier

\subsection{Additional plots for the microbiome application}
\FloatBarrier
\textbf{Visual description of the dataset}
\begin{figure}[ht!]
    \centering
    \includegraphics[scale=.5]{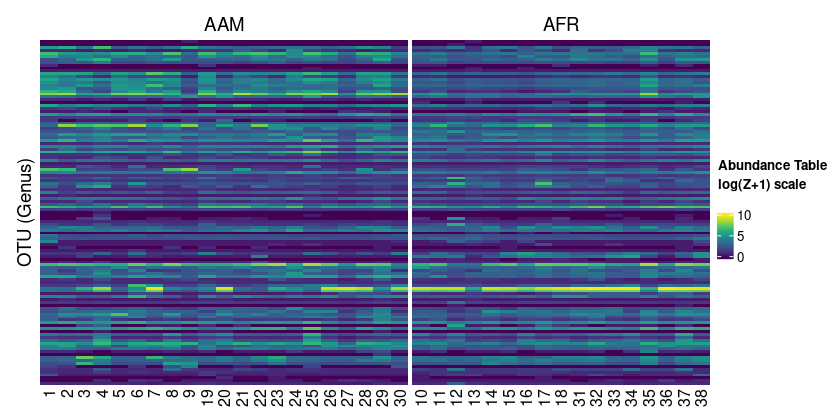}
    \caption{Heatmap of the considered abundance table. The OTUs (at the Genus level) are reported by row, while the columns indicate the subjects, divided by nationality. The count data are transformed as $\log(\bm{Z}+1)$.}
    \label{fig:heatmap}
\end{figure}
\textbf{Percentage of abundance classes per OTU}
\begin{figure}[ht!]
    \centering
    \includegraphics[scale=.34]{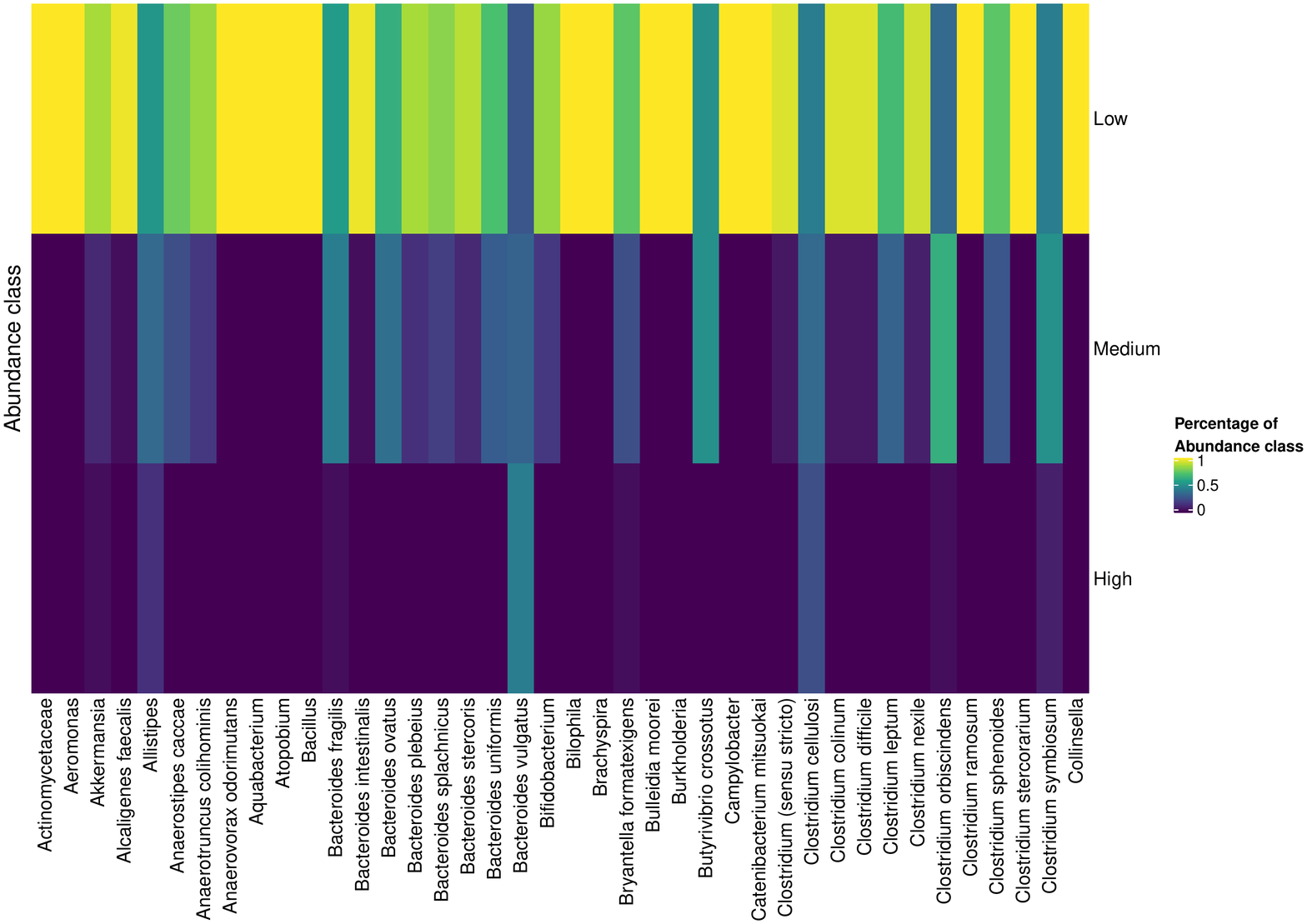}
    \caption{Distribution of the three estimated abundance classes - Part I}
    \label{AB0}
\end{figure}
\begin{figure}[ht!]
    \centering
        \includegraphics[scale=.34]{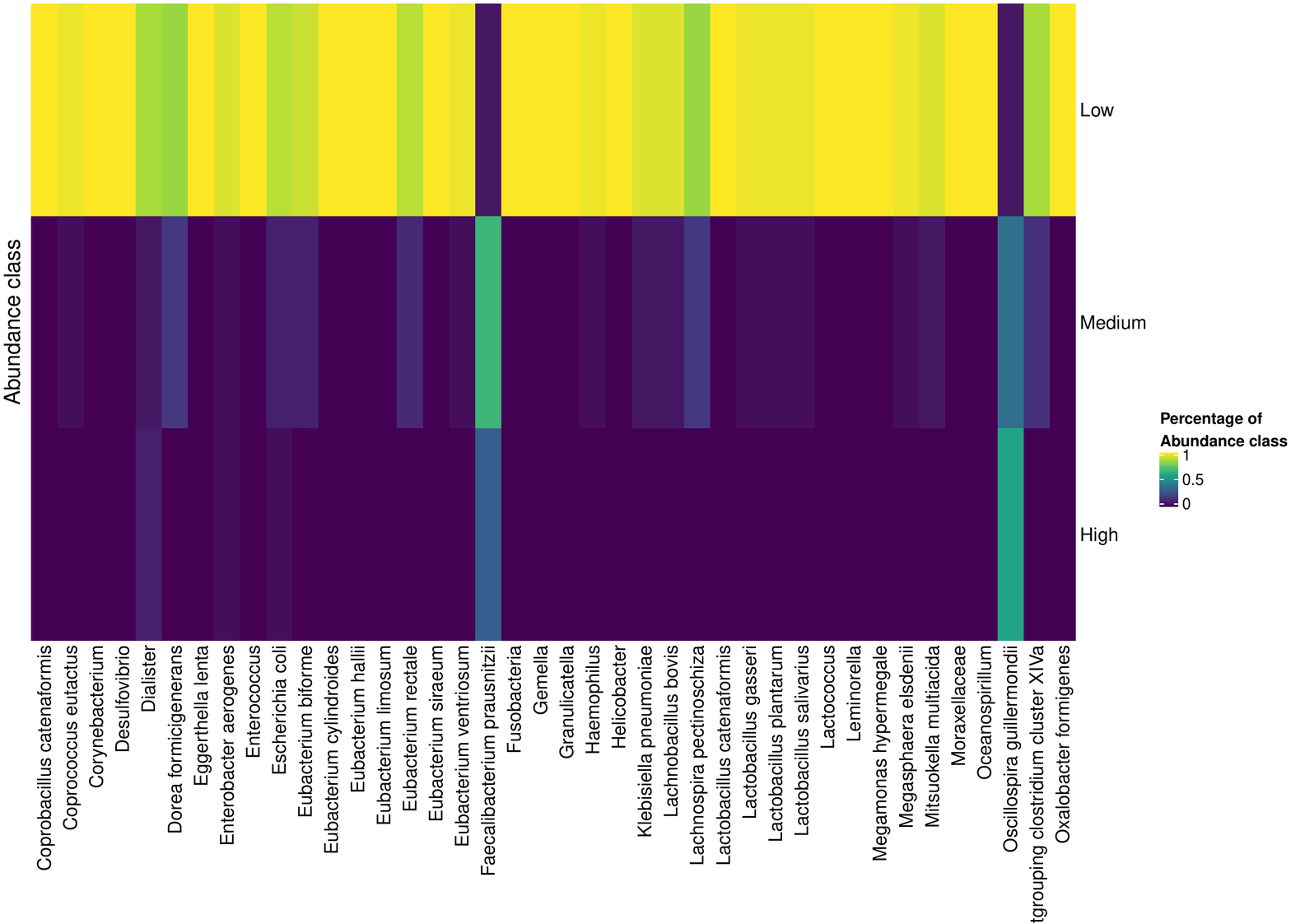}
\caption{Distribution of the three estimated abundance classes - Part II}
    \label{AB1}
\end{figure}

\begin{figure}[ht!]
    \centering
    \includegraphics[scale=.34]{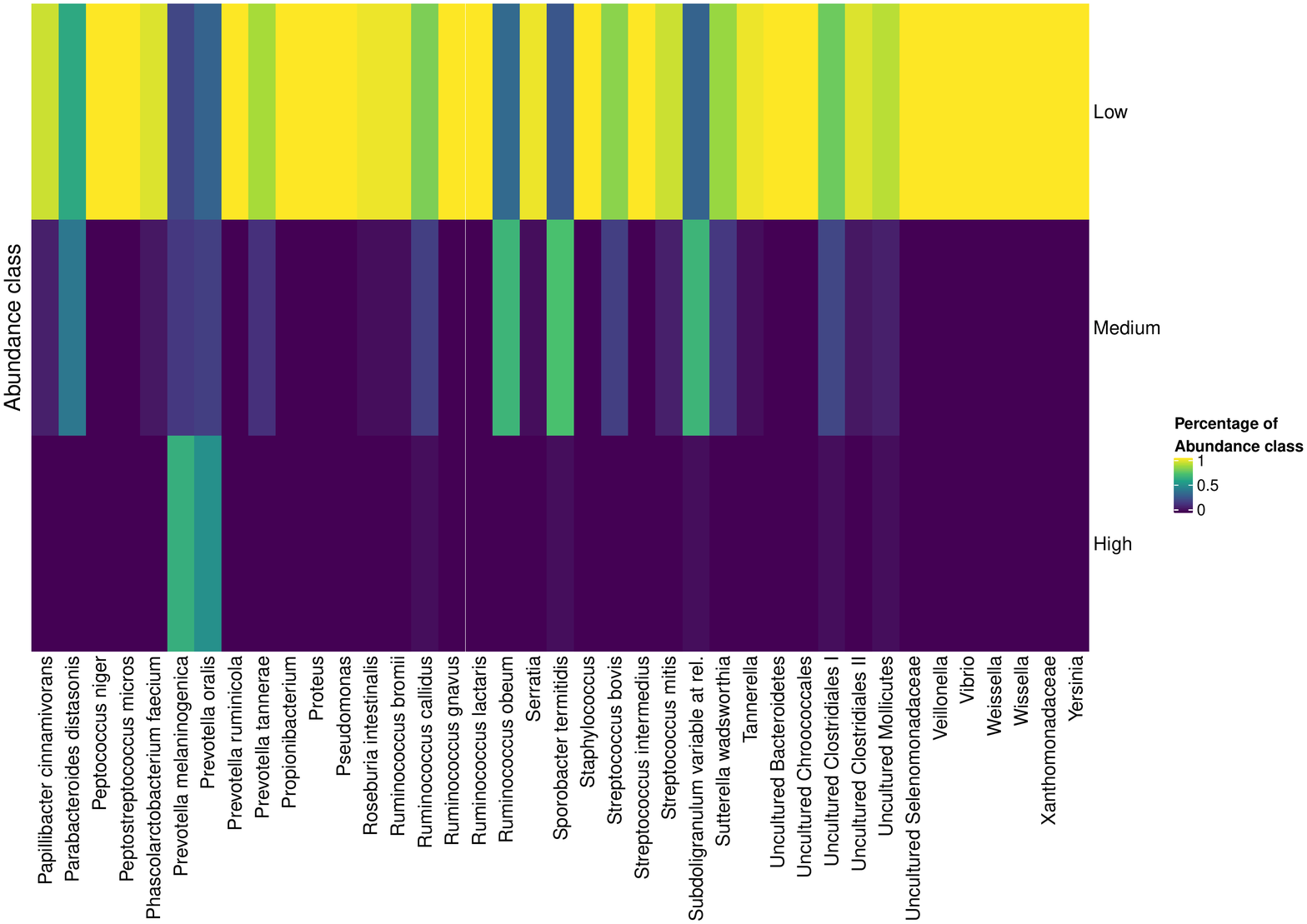}
\caption{Distribution of the three estimated abundance classes - Part III}    \label{AB2}
\end{figure}
\FloatBarrier

\clearpage
\textbf{Boxplots of the distributional characteristics across DCs}
\begin{figure}[!htpb]
	\centering 
	\includegraphics[scale=.35]{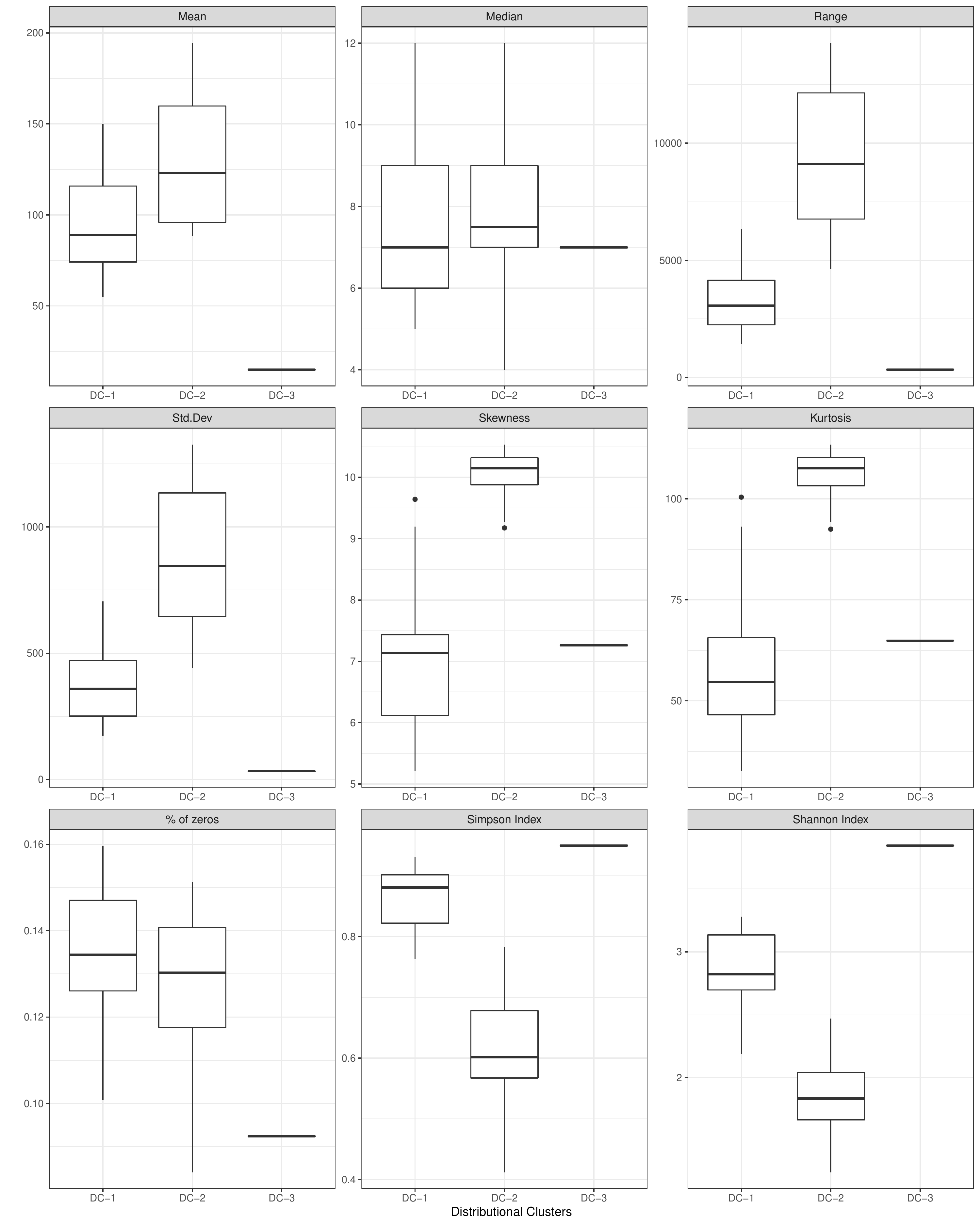}
	\caption{Boxplots representing how mean, median, range, standard deviation, skewness, kurtosis, \% of zeros, Shannon index, and Simpson index of each microbiome are distributed across the DCs. The plots highlight different distributional differences among the three DCs.}
	\label{Boxplots}
\end{figure}	

\FloatBarrier

\clearpage
\bibliographystyle{plainnat}

\bibliography{ALLBIB}
\end{document}